\documentclass[11pt]{amsart}

\usepackage{authblk}

\usepackage[pdftex]{graphicx}
\usepackage[ddmmyyyy]{datetime}

\usepackage{amssymb,amsfonts,amsmath}
\usepackage{amsthm,xypic,enumerate}
\usepackage{color,tikz,pgf}
\usepackage{fancyhdr}
\input xy
\xyoption{all}
\usepackage{hyperref}

\newcommand{\R}{\mathbb{R}}

\newcommand{\mmS}{\mathcal{S}}
\newcommand{\mmY}{\mathcal{Y}}
\newcommand{\oS}{\overline{S}}
\newcommand{\oE}{\overline{E}}
\newcommand{\oF}{\overline{F}}

\newcommand{\self}{\scriptscriptstyle \circlearrowleft}

\DeclareMathOperator*{\diag}{diag}
\DeclareMathOperator*{\Sub}{Sub}\newcommand{\Subn}{\Sub\nolimits}
\DeclareMathOperator*{\Enz}{Enz}\newcommand{\Enzn}{\Enz\nolimits}
\DeclareMathOperator*{\Int}{Int}\newcommand{\Intn}{\Int\nolimits}
\DeclareMathOperator*{\Rct}{Rct}
\DeclareMathOperator*{\Con}{Con}

\newenvironment{enumerate*}[1][{}]{\begin{itemize}\renewcommand{\labelitemi}{#1}}{\end{itemize}}

\numberwithin{equation}{section}
\newtheorem{theorem}[equation]{Theorem}
\newtheorem{proposition}[equation]{Proposition}
\newtheorem{corollary}[equation]{Corollary}
\newtheorem{lemma}[equation]{Lemma}

\theoremstyle{definition}
\newtheorem{definition}[equation]{Definition}

\title{Variable elimination in post-translational modification reaction networks with mass-action kinetics}

\author{Elisenda Feliu, Carsten Wiuf}

\thanks{{\it Authors affiliation}: Bioinformatics Research Centre, Aarhus University, C. F. M\o llers All\'{e} 8, DK-8000 Aarhus, Denmark}
\thanks{{\it Corresponding author}: Elisenda Feliu, efeliu@birc.au.dk}

\date{\today}

\begin{document}

\maketitle

\begin{abstract}
We define a subclass of  Chemical Reaction Networks  called Post-Translational Modification systems. Important biological examples of such systems include MAPK cascades and two-component systems which are well-studied experimentally as well as theoretically. The steady states of such a system are solutions to a system of polynomial equations with as many variables as equations. Even for small systems the  task of finding the solutions is daunting. We  develop a mathematical framework based on the notion of a \emph{cut}, which provides a linear elimination procedure to reduce the number of variables in the system. 
The steady states are parameterized algebraically by a set of ``core'' variables, and the non-negative steady states correspond to non-negative values of the core variables.
 Further, minimal cuts are the  connected components in the species graph and provide conservation laws. A criterion for when a set of independent conservation laws can be derived from cuts is given.

\end{abstract}

\section{Introduction}

Signaling systems  play an important role in  regulation of cellular processes and are essential for cellular decision making. Typical signaling systems react to stimulus in the (cellular) environment and transmit a signal through connected layers of biochemical species. The layers provide means to adjust the response according to the stimulus. A common form of signaling systems is Post-translational Modification (PTM) systems where species are  activated  in  chemical reactions in order to propagate the signal through the system.

PTM systems have attracted considerable theoretical attention  due to their abundance in nature \cite{Huang-Ferrell} and regular form \cite{TG-rational}. The dynamics can be modeled as $\frac{dx(t)}{dt}=p(x)$, where $x=(x_1,\ldots,x_n)$ are the variables (concentrations of species) of the system and $p(x)$ is a vector of polynomials in $x$. Only certain types of  reactions are allowed, restricting the form of $p(x)$.  In particular, small specific systems have been scrutinized, focusing on  the dynamical behavior and the steady states of the systems. Examples include the biologically important  MAPK cascades \cite{Huang-Ferrell,Kholodenko-fourdim,Markevich-mapk}, as well as simpler signaling cascades  \cite{Feliu:2010p94,Heinrich-kinase,Ventura-Hidden}. 

We focus on the  steady states of a PTM system (defined formally in the next section) and how to determine them. 
 Taken with mass-action kinetics, the system's steady states are solutions to a set of polynomial equations in the species and  with coefficients given by unknown kinetic rates  (i.e.~unspecified parameters). In particular, the number of equations to be solved is equal to the number of species. Even small systems might have many variables such that analytical solutions are difficult to obtain and numerical solutions are prone to errors. Further, many PTM systems admit multistationarity  (the existence of more than one  steady state under particular biological conditions) which is  a mechanism for cellular decision making \cite{TG-Nature}. It is therefore of interest to determine the parameters for which mono- and multistationarity occur. Several non-necessary conditions for a unique positive steady state are known \cite{angelisontag2,craciun-feinbergI,feinberg-def0}, but when these fail, multistationarity is difficult to determine and often decided based on a random parameter search. Procedures to eliminate variables (hence, equations)  is therefore fundamental to the theoretical understanding of these systems as well as for numerical analysis.

 Our work is inspired  by previous work by Thomson and Gunawardena (TG) \cite{TG-rational} which we extend to embrace a range of  important PTM systems such as signaling cascades (including the MAPK cascade) and two-component systems with phosphorelays and phosphotranfer \cite{krell}, as well as systems with self-interactions.  We develop the idea of a \emph{cut} $\mmS_{\alpha}$, a subset of the substrates $\mmS$ with certain properties that allow us to express the steady state equations as rational functions in the ``core'' variables  $\mmS\setminus\mmS_{\alpha}$, providing an algebraic parameterization of the steady states in terms of the core variables. If the core variables take positive values at steady state, then we show that all other concentrations are either zero or positive as well. 
 
Further, we show that cuts relate to \emph{conservation laws} (conserved quantities that imply that the dynamics takes place in an affine invariant subspace of $\mathbb{R}^n$) that arise as connected components in the species graph \cite{angelisontag2}. Conservation laws are often used as a first step to reduce the dimensionality of the system. In our approach, conservation laws come into play after elimination of  variables from the steady state equations. In this way, we  allow for a larger reduction in the number of core variables. 

Our appoach makes use of algebraic tools as well as some basic graph properties; for example Tutte's Matrix-Tree theorem \cite{Tutte-matrixtree,TG-rational}.  One benefit is that parameters are treated as symbolic constants and do not need to be fixed or assumed known in advance. This is particularly relevant in biology, where we often are faced with systems that depend on experimental parameters (kinetic rates), which are difficult to determine.

\section{Post-translational modification systems}
\label{results}

\subsection{PTM system}
\label{PTMs}

A \emph{post-translation modification (PTM) system} consists of two non-empty sets of species, $\mathcal{S} = \{S_{1},\dots,S_{N}\}$ (the \emph{substrates}) and $ \mathcal{Y} = \{Y_{1},\dots,Y_{P}\}$ (the \emph{intermediate complexes}) with $\mathcal{S}\cap \mathcal{Y}=\emptyset$, and a set of reactions $\Rct=R_{a} \cup R_{b} \cup R_{c}\cup R_{d}$ with associated positive reaction rate constants:
\vspace{-0.2cm}
\begin{align*}
R_{a} &= \{ S_i+S_j \xrightarrow{a_{i,j}^k} Y_{k}| (i,j,k)\in I_{a} \}  & R_{c} &=  \{Y_{i} \xrightarrow{c_{i,j}} Y_{j} | (i,j)\in I_{c},i\neq j\}   \\
R_{b} &=  \{ Y_{k} \xrightarrow{b_{i,j}^k} S_{i}+S_{j} | (i,j,k)\in I_{b} \} &
R_{d} &=  \{S_{i} \xrightarrow{d_{i,j}} S_{j} | (i,j)\in I_{d},i\neq j\}
\end{align*}
for $I_{a}, I_{b} \subseteq  \{1,\dots,N\}^2\times \{1,\dots,P\}$, $I_c\subseteq \{1,\dots,P\}^2$ and $I_{d} \subseteq \{1,\dots,N\}^2$.
To fix the notation, we assume that any $(i,j,k)\in I_a\cup I_b$ satisfies $i\leq j$, so that self-interactions  a priori are allowed. If the rate constants are not required, we put an arrow to indicate a reaction and omit the rates.
Further:
\begin{enumerate}[(i)]
\item All chemical species are involved in at least one reaction.
\item For every intermediate complex $Y_k$ there exist $i\leq j$, indices $k_1,\dots,k_r$ and a chain of reactions
$Y_k \rightarrow Y_{k_1}\rightarrow \dots \rightarrow Y_{k_r} \rightarrow S_i+S_j. $
\end{enumerate}
Assumption (ii) ensures that $Y_k$ \emph{ultimately dissociates into two substrates}. Also, we allow that there are more than one $Y_k$ such that $S_i+S_j\rightarrow Y_k$  or $Y_k\rightarrow S_i+S_j$ for given $S_i, S_j$. 
For convenience, we put $c_{i,j}=0$, $d_{i,j}=0$ if $(i,j)\notin I_{c}$ or $I_{d}$ respectively, and similarly $a_{i,j}^k=0$ and $b_{i,j}^k=0$ if $(i,j,k)\notin I_a$ or $I_b$, respectively. For $i\leq j$ and $k$, we define $a_{j,i}^k=a_{i,j}^k$ and $b_{j,i}^k=b_{i,j}^k$.  For later use, we define
$$\mmS_{\self}=\{S_i\in \mmS|(i,i,k)\in I_a\cup I_b\textrm{ for some }k \}$$
to be the set of self-interacting substrates.

This setting fits post-translational modification of proteins catalyzed by enzymes as well as the transfer of modifier groups:
$$\xymatrix@C=15pt{  E + S \ar@<0.5ex>[r] & Y \ar[r] \ar@<0.5ex>[l] &  E+S^*} \qquad\xymatrix@C=15pt{  P^* + S \ar@<0.5ex>[r] & Y \ar@<0.5ex>[r] \ar@<0.5ex>[l] &  P+S^* \ar@<0.5ex>[l] }$$
where $S^*,P^*$ are modified proteins (substrates), $S, P$ their corresponding unmodified forms, $E$ an enzyme (substrate) and $Y$ an intermediate complex. That is, ${\mmS}=\{S,S^*,P,P^*,E\}$ and ${\mmY}=\{Y\}$. In the first case, the  attachment of the modifier group is catalyzed by the enzyme $E$, whereas in the second case, a modifier group is transferred from $P^*$ to $S$.
Modification of a substrate or an intermediate complex without the involvement of other species is modeled by  $S\rightarrow S^*$ and $Y\rightarrow Y^*$, respectively.

As an  example consider the PTM system with $\mmS=\{S_1,S_2,S_3,S_4,S_5\}$, $\mmY=\{Y_1,Y_2,Y_3\}$ and reactions
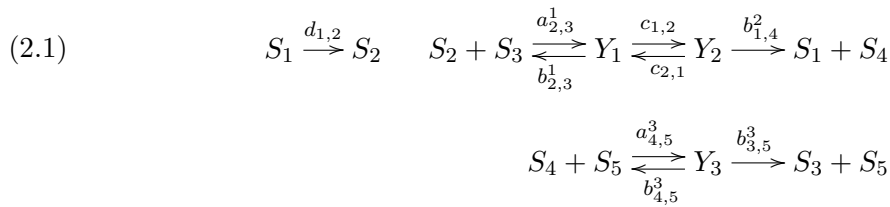
\begin{align}\label{mainex} 
\xymatrix@C=15pt{  S_1 \ar[r]^{\small d_{1,2}} & S_2}\quad\xymatrix@C=20pt{ S_2+ S_3 \ar@<0.5ex>[r]^(0.6){a_{2,3}^1} & Y_1 \ar@<0.5ex>[r]^{c_{1,2}} \ar@<0.5ex>[l]^(0.4){b_{2,3}^1} & Y_2 \ar[r]^(0.4){b_{1,4}^2} \ar@<0.5ex>[l]^(0.4){c_{2,1}} &  S_1+S_4  } \\
\xymatrix@C=20pt{  S_4+S_5 \ar@<0.5ex>[r]^(0.6){a_{4,5}^3} & Y_3 \ar[r]^(0.35){b_{3,5}^3} \ar@<0.5ex>[l]^(0.35){b_{4,5}^3}&  S_3+S_5} \nonumber
\end{align}
One interpretation is that $S_1$ is modified to $S_2$.  The modifier group is then transferred from $S_2$ to $S_3$ with the formation of two intermediate complexes $Y_1,Y_2$, causing the modification of $S_3$ to $S_4$ and the demodification of $S_2$ to $S_1$. Finally, $S_4$ is demodified via a Michaelis-Menten mechanism catalyzed by an enzyme $S_5$.

{\bf Nomenclature.} We introduce a few concepts that will be used in the following, some of which are taken from Chemical Reaction Network Theory (CRNT)  \cite{feinbergnotes,feinberg-invariant}. Consider the set of \emph{complexes} of the reaction system:
$$\mathcal{C}=\mmY \cup \{S_i,\ S_j |\ (i,j)\in I_{d} \}\cup  \{S_i+S_j |\ (i,j,k)\in I_{a}\cup I_b \textrm{ for some }k \}.$$ 
 Then:
\begin{enumerate}[$\bullet$]
\item $A\in \mathcal{C}$ \emph{reacts to}  $B\in\mathcal{C}$ if there exists a reaction $A\rightarrow B$. 
\item  $A\in \mathcal{C}$ \emph{ultimately reacts to}  $B\in\mathcal{C} $  if there exists a sequence of reactions $A \rightarrow A_{1} \rightarrow \dots \rightarrow A_{r} \rightarrow B$
with $A_{m}\in \mathcal{C}$. If  $A_{m}\in \widetilde{\mmY}\subseteq\mmY$ for all $m$, then  \emph{$A$ ultimately reacts to $B$ via $\widetilde{\mmY}$}.
\item $S_i$ and $S_j$  \emph{interact} if for some $Y_k$ either  $S_i+S_j$  reacts to $Y_k$ or vice versa. 
\item $S_i,S_j$ are \emph{1-linked}  if $d_{i,j}$ or $d_{j,i}\neq 0$. $Y_k,Y_v$ are 1-linked if  $c_{k,v}$ or $c_{v,k}\neq 0$. $S_i$ and $Y_k$ are 1-linked if for some $j$, $S_i+S_j$ reacts to $Y_k$ or vice versa ($j=i$ is allowed).
\end{enumerate}

Assumption (ii) of a PTM system ensures that all intermediate complexes ultimately react to some $S_i+S_j$ via $\mmY$.

\subsection{Mass-action kinetics} The set of reactions together with their associated rate constants give rise to a polynomial system of ordinary differential equations taken with \emph{mass-action kinetics}:
\begin{align*}
\dot{Y_{k}} &=  \sum_{j=1}^N \sum_{i=1}^j (a_{i,j}^k S_{i}S_j - b_{i,j}^kY_{k}) + \sum_{v=1}^P (c_{v,k}Y_{v}-c_{k,v}Y_{k}), &  k=1,\dots,P, \\
\dot{S_{i}}&=  \sum_{j=1}^N \sum_{k=1}^P \epsilon_{i,j}(-  a_{i,j}^k S_iS_j +b_{i,j}^k Y_{k})+  \sum_{j=1}^N (d_{j,i}S_j - d_{i,j}S_i), &i=1,\dots,N,
\end{align*}
where $\epsilon_{i,j}=1$ if $i\neq j$ and $2$ if $i=j$ and where $\dot{x}$ denotes $dx/dt$ for $x=x(t)$. Here we abuse notation and let $S_i,Y_k$ denote the concentrations of the species $S_i,Y_k$ as well. 
The steady states are the solutions to the polynomial system  obtained by setting the derivatives to zero, i.e. $\dot{Y_{k}}=0$ and $\dot{S_{i}}=0$:
\begin{align}
0 =&   \sum_{j=1}^N \sum_{i=1}^j (a_{i,j}^k S_{i}S_j - b_{i,j}^kY_{k}) + \sum_{v=1}^P (c_{v,k}Y_{v}-c_{k,v}Y_{k}), \label{dY1}  & k=1,\dots,P,\\
0 =& \sum_{j=1}^N \sum_{k=1}^P \epsilon_{i,j}(-  a_{i,j}^k S_iS_j +b_{i,j}^k Y_{k})+  \sum_{j=1}^N (d_{j,i}S_j - d_{i,j}S_i), & i=1,\dots,N.
  \label{dS1} 
\end{align}
This system is quadratic in the variables $Y_k, S_i$, but the only quadratic terms are of the form $S_iS_j$. It is linear in $Y_k$.

It is convenient to treat  the reaction rate constants as parameters with unspecified (positive) values and view $a_{i,j}^k, b_{l,m}^r, c_{u,v},d_{w,t}$ as symbols. For that, let 
$$\Con=\{a_{i,j}^k| (i,j,k)\in I_{a}\}\cup \{b_{i,j}^k|  (i,j,k)\in I_{b}\}\cup \{ c_{k,v}|  (k,v)\in I_{c}\}\cup \{d_{k,v}| (k,v)\in I_d\}$$
 be the set of the non-zero parameters (symbols). Then, the system \eqref{dY1}-\eqref{dS1} is quadratic in $\mmS\cup\mmY$ with coefficients in the field $\R(\Con)$. Further, if all $S_i$ are considered part of the coefficient field, then the system is linear with coefficients in $\R(\Con\cup\, \mathcal{S})$ and variables $Y_1,\dots,Y_P$.

Only  \emph{non-negative solutions} of the steady state equations are biologically meaningful. To study positivity of solutions, we introduce the concept of \emph{S-positivity}. Let $X=\{x_1,\dots,x_r\}$ be a finite set. A non-zero polynomial  in $\R[X]$ with  non-negative coefficients is called \emph{S-positive}.
Similarly, a rational function $f$ is S-positive if it is a quotient of two S-positive polynomials. If  $x_{1},\dots,x_{r}$ are substituted by positive real numbers in $f$, we obtain  a positive real number. In general, a rational function $f=p/q$ in  $z_1,\dots,z_s$ and  coefficients in  $\R(X)$ is S-positive if  the coefficients of $p$ and $q$ are S-positive rational functions in $x_1,\dots,x_r$. If $f$ is a rational function in $x_{1},\dots,x_{r}$ and $x_{i}=g(x_{1},\dots,\widehat{x_{i}},\dots,x_{r})$ with $g$ a rational function, then substituting $g$ into $f$ gives $f$ as a rational function in $x_{1},\dots,\widehat{x_{i}},\dots,x_{r}$.
 
The differential equations of  Example \eqref{mainex} are:

\vspace{-0.2cm}
\noindent
\begin{minipage}{.52\textwidth}
\begin{align}
\dot{Y_1} &= a_{2,3}^1 S_2S_3  -(b_{2,3}^1+c_{1,2})Y_1 + c_{2,1} Y_2 \label{maindiffeqs} \\
\dot{Y_2} &= c_{1,2}Y_1 - (b_{1,4}^2+ c_{2,1}) Y_2 \nonumber  \\
\dot{Y_3} &= a_{4,5}^3 S_4S_5 - (b_{4,5}^3  +b_{3,5}^3) Y_3 \nonumber  \\
\dot{S_5} &= -a_{4,5}^3 S_4S_5 + (b_{4,5}^3  +b_{3,5}^3) Y_3 \nonumber
\end{align}
\end{minipage}
\begin{minipage}{.45\textwidth}
\begin{align*}
\dot{S_1} &= -d_{1,2}S_1 + b_{1,4}^2Y_2   \\
\dot{S_2} &=d_{1,2}S_1 -a_{2,3}^1 S_2S_3 + b_{2,3}^1 Y_1  \\
\dot{S_3} &= -a_{2,3}^1 S_2S_3  + b_{2,3}^1 Y_1  +b_{3,5}^3 Y_3 \\ 
\dot{S_4} &=-a_{4,5}^3 S_4S_5+  b_{1,4}^2Y_2  + b_{4,5}^3Y_3. 
\end{align*}
\end{minipage}

\medskip\noindent
To compute the steady states, we can use $\dot{Y_3}=0$ to eliminate  $Y_3$  as a  function of the substrates. Also $Y_1,Y_2$ can  be eliminated by solving the linear system $\dot{Y_1}=\dot{Y_2}=0$. This is a general feature of PTM systems and is covered in Section \ref{eliminter}.

Further, observe that $\dot{S_5}+\dot{Y_3}=0$, which implies that the sum $S_5+Y_3$ is independent of time and thus conserved. In fact,  it implies that one of the equations $\dot{Y_3}=0$ and $\dot{S_5}=0$ is redundant.  Removing one of them leaves a polynomial system with $7$ equations in $8$ variables, and thus the solutions to the steady state equations form an algebraic variety of dimension at least one. This redundancy can be compensated for by fixing the value $S_5+Y_3=\oS$  and adding this relation to the steady state equations.

In the next section we discuss the existence of the so-called \emph{conservation laws} and provide a graphical procedure to determine (some of)  them. In most cases the procedure provides a set of independent conservation laws, but, as will be discussed below, this might not  always be the case.

\subsection{Conservation laws}\label{sec:conslaws}
We consider  systems where   inflow of species is not allowed and  species are not degraded or able to diffuse out. Such systems are ``entrapped'' in contrast to open systems (so-called ``continuous flow stirred tank reactors'')  \cite{craciun-feinberg-semiopen}. PTM systems are entrapped and have conservation laws  that reflect that the total amount of species remains constant either in free form $S_i$ or in bounded form $Y_j$. These laws follow from the system of differential equations and appear  as linear combinations of species (e.g. $S_5+Y_3=\oS$ in the example above).

The existence of conservation laws implies that the dynamics of the system takes place in a proper invariant subspace of $\R^{N+P}$. We identify $\R^{N+P}$ with the real vector space generated by $\mmS\cup \mmY$ so that $\R^{N+P} \equiv \langle S_1,\dots,S_N,Y_1,\dots,Y_P\rangle$. The species $S_i$ and $Y_k$ are unit vectors with a one in the $i$-th and $(N+k)$-th entry, respectively, and all other entries being zero.
A vector $v=(\lambda_1,\dots,\lambda_N,\mu_1,\dots,\mu_P)$ is identified with the   linear combination of species $\sum_i \lambda_i S_i + \sum_k \mu_k Y_k$.

Consider the \emph{stoichiometric subspace} of $\R^{N+P}$ \cite{craciun-feinberg-semiopen} of a PTM system:
$$\Gamma = \langle S_i+S_j-Y_k |\ (i,j,k)\in I_a\cup I_b   \rangle + \langle Y_k-Y_v|\  (k,v)\in I_c\rangle+ \langle S_i-S_j|\  (i,j)\in I_d\rangle.   $$
If $(\lambda_1,\dots,\lambda_N,\mu_1,\dots,\mu_P)\in \Gamma^{\perp}$, then 
 $\sum_i \lambda_i \dot{S}_i + \sum_k \mu_k \dot{Y}_k=0$. The converse might not be true \cite{feinberg-invariant}. It follows that any basis $\{\omega^1,\dots,\omega^d\}$ of $\Gamma^{\perp}$ provides a set of independent conserved quantities $\sum_{i=1}^N \lambda_i^l S_i + \sum_{k=1}^P \mu_k^l Y_k$  if $\omega^l=(\lambda_1^l,\dots,\lambda_N^l,\mu_1^l,\dots,\mu_P^l)$. Therefore, if \emph{total amounts} $\oS_1,\dots,\oS_d\in \R_{+}$ are provided, we require the steady state solutions to satisfy:
 \begin{equation} \label{totalamounts}
\oS_l = \sum_{i=1}^N \lambda_i^l S_i + \sum_{k=1}^P \mu_k^l Y_k  \qquad l=1,\dots,d.
\end{equation}
Total amounts  are fixed by the initial concentrations of the species.
 We say that
equations \eqref{totalamounts} are independent if the system has maximal rank, or equivalently, if the corresponding vectors of $\Gamma^{\perp}$ are independent.

We introduce the concepts of  a cut and a non-interacting graph and show that they provide means to obtain conservation laws.

\begin{definition}\label{def1}
 Let a non-empty set ${\mmS}_{\alpha}\subseteq{\mmS}$ be given and let the associated set ${\mmY}_{\alpha}\subseteq{\mmY}$ be the smallest set such that $Y_k\in{\mmY}_{\alpha}$ if $Y_k$  is 1-linked to some $S_i\in{\mmS}_{\alpha}$ or to $Y_m\in{\mmY}_{\alpha}$. 
\begin{itemize}
\item[(i)]  ${\mmS}_{\alpha}$ is \emph{closed} if  $S_j$ belongs to ${\mmS}_{\alpha}$ whenever $S_i\in {\mmS}_{\alpha}$ is 1-linked to $S_j$, and if $S_i$ and $S_j$ interact and are 1-linked to   $Y_k\in {\mmY}_{\alpha}$, then $S_i$ or $S_j$ are in $\mmS_{\alpha}$. 
\item[(ii)]  ${\mmS}_{\alpha}$ is a \emph{cut}  if (a)  $S_i,S_j\in{\mmS}_{\alpha}$ do  not interact for any $i, j$,  and (b) ${\mmS}_{\alpha}$ is closed.
\item[(iii)] A cut ${\mmS}_{\alpha}$ is \emph{minimal} if it has no proper closed subsets.
\end{itemize}
\end{definition}
Condition (ii) implies that a self-interacting substrate $S\in \mmS_{\self}$ cannot belong to any cut, that is, $\mmS_{\alpha}\cap \mmS_{\self}=\emptyset$ for any cut $\mmS_{\alpha}$. 
Note that a closed subset $\mmS'$ of a cut is also a cut. The union of two disjoint cuts $\mmS_{\alpha},\mmS_{\alpha}'$ is a cut if $\mmY_{\alpha}\cap \mmY_{\alpha}'=\emptyset$. 

In the PTM system with reactions 
$\xymatrix@C=12pt{ S_1+ S_4  \ar@<0.3ex>[r] & Y_2 \ar@<0.3ex>[l] \ar@<0.3ex>[r]  &   S_2 + S_4 \ar@<0.3ex>[l] }$ and  $\xymatrix@C=12pt{  Y_1 \ar@<0.3ex>[r]   &  S_2+S_3,  \ar@<0.3ex>[l]}$ the set
 $\{S_1,S_2\}$ is a cut, while   $\{S_1,S_3\}$ is not. There are no proper closed subsets of $\{S_1,S_2\}$ and thus the cut is minimal.

\begin{definition}\label{def2}
 Let a non-empty set ${\mmS}_{\alpha}\subseteq{\mmS}$ be given and let  ${\mmY}_{\alpha}\subseteq{\mmY}$  be as in Definition~\ref{def1}. Further, let $G_{\mmS_{\alpha},\mmY_{\alpha}}$ be the graph with node set ${\mmS}_{\alpha}\cup{\mmY}_{\alpha}$ and edges between 1-linked nodes. The graph is \emph{non-interacting} if  it is connected and   ${\mmS}_{\alpha}$ is a cut.
\end{definition}

If $\mmS_{\alpha}=\mmS$, then $\mmY_{\alpha}=\mmY$. All graphs $G_{\mmS_{\alpha},\mmY_{\alpha}}$ are naturally subgraphs of $G_{\mmS,\mmY}$.
Without proof we state the following:

\begin{lemma}\label{cut_graph}
 Let ${\mmS}_{\alpha}$ be a cut and $G'$ be a connected subgraph of $G_{\mmS_{\alpha},\mmY_{\alpha}}$ with node set $\mmS'\cup\mmY'$, $\mmS'\subseteq\mmS_{\alpha}$ and $\mmY' \subseteq\mmY_{\alpha}$. The following are equivalent:
\begin{itemize}
\item[(i)] $\mmS'$ is closed with associated set $\mmY'$.
\item[(ii)] $G'$ is a connected component of $G_{\mmS_{\alpha},\mmY_{\alpha}}$.
\item[(iii)] $G'$ is non-interacting and  contains only species in $\mmS_{\alpha}\cup \mmY_{\alpha}$.
\end{itemize}
If either is the case, then $\mmS'$ is a minimal cut and $G'=G_{\mmS',\mmY'}$. 
\end{lemma}

Thus, the non-interacting graphs containing substrates only in a cut $\mmS_{\alpha}$ are exactly the connected components of $G_{\mmS_{\alpha},\mmY_{\alpha}}$.
All  non-interacting graphs contain some node from $\mmS$ (condition (ii) of a PTM system).  However, such a graph might not exist. Consider for example the system with reactions
$ \xymatrix@C=12pt{S_1 \ar@<0.3ex>[r] & S_3, \ar@<0.3ex>[l]}$ $\xymatrix@C=12pt{S_2\ar@<0.3ex>[r] & S_3,   \ar@<0.3ex>[l]}$ $\xymatrix@C=12pt{S_1+S_2\ar@<0.3ex>[r] & Y_1.  \ar@<0.3ex>[l]}$
The graph $G_{\mmS,\mmY}$ is
\begin{center}
\begin{tikzpicture}
\node (Y1) at (1,0) {$Y_1$};
\node (S1) at (2,0.5) {$S_1$};
\node (S3)at (3,0) {$S_3$};
\node (S2) at (2,-0.5) {$S_2$};
\path[draw] (S1) to (Y1);
\path[draw] (S2) to (Y1);
\path[draw] (S1) to (S3);
\path[draw] (S2) to (S3);
\end{tikzpicture}
\end{center}
Condition (b) of Definition~\ref{def1}(ii) implies that any non-interacting graph must contain all four species, which contradicts  condition (a) of the same definition.

\begin{lemma}\label{lemma1}
Let $H_1,\dots,H_n$ be the non-interacting graphs of a PTM system, $C_l$ the node set of  $H_l$, $\mmS_l=\mmS\cap C_l$ and $\mmY_l= \mmY \cap C_l$. Then, $\dot{\omega_l}=0$ for 
$$\omega_l=\sum_{S\in \mmS_l } S + \sum_{Y\in  \mmY_l} Y\qquad l=1,\dots,n.  $$ 
That is, $H_l$ corresponds to a conservation law and $\omega_l$ is fixed by the initial amounts.
\end{lemma}
\begin{proof}
Substrates in $C_l$ interact only with substrates in $\mmS\setminus \mmS_l$ and thus, by  definition of $\mmY_l$, if $a_{i,j}^k\neq 0$ or $b_{i,j}^k\neq 0$ for $i\neq j$ then: (a) if $S_i$ (resp. $S_j$) is in $\mmS_{l}$, then $S_j$ (resp. $S_i$) belongs to $\mmS\setminus\mmS_l$, and $Y_k\in \mmY_{l}$; (b) if $Y_k\in \mmY_{l}$, then either $S_i$ or $S_j$, but not both, belongs to $\mmS_l$.
If $c_{v,k}\neq 0$ or $c_{k,v}\neq 0$, then $Y_k,Y_v$ belong to the same non-interacting graph (if any); if $d_{i,j}\neq 0$ or $d_{j,i}\neq 0$, then $S_i,S_j$ belong to the same non-interacting graph (if any). Since $\mmS_l\cap \mmS_{\self}=\emptyset$  for $Y_k\in \mmY_l$ and $S_i\in \mmS_l$ we have:
\begin{align*}
\dot{Y_{k}} &=  \sum_{i| S_i\in \mmS_l} \sum_{j| S_j\in \mmS\setminus \mmS_l} (a_{i,j}^k S_{i}S_j - b_{i,j}^kY_{k}) +  \sum_{v|Y_v\in \mmY_l} (c_{v,k}Y_{v}-c_{k,v}Y_{k})  \\
\dot{S_{i}}&=   \sum_{k|Y_k\in \mmY_l} \sum_{j| S_j\in \mmS\setminus \mmS_l} (-  a_{i,j}^k S_iS_j +b_{i,j}^k Y_{k})+   \sum_{j| S_j\in \mmS_l} (d_{j,i}S_j - d_{i,j}S_i).
\end{align*}
 It follows that $\sum_{k|Y_k\in \mmY_l}\sum_{v|Y_v\in \mmY_l} (c_{v,k}Y_{v}-c_{k,v}Y_{k})=0$ and $ \sum_{i| S_i\in \mmS_l} \sum_{j| S_j\in \mmS_l} (d_{j,i}S_j - d_{i,j}S_i)=0$. Similarly, the remaining terms in $\dot{\omega}_l$ cancel. Thus, $\dot{\omega_l}=0$.
\end{proof}

Thus, each non-interacting graph gives rise to a conserved amount.
If each non-interacting graph contains a species that only belongs to that specific graph, then  the $\omega_l$'s are independent. In particular, conservation laws derived from the connected components of $G_{\mmS_{\alpha},\mmY_{\alpha}}$ for some cut $\mmS_{\alpha}$ are independent. In general,  the set of conservation laws found from Lemma~\ref{lemma1} can be reduced to a set of  independent conservation laws. 

In Example \eqref{mainex}, the graph $G_{\mmS,\mmY}$ is 
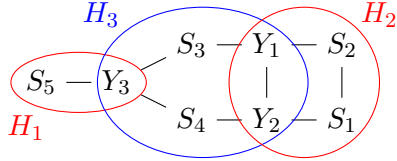
\begin{center}
\begin{tikzpicture}
\node (S5) at (0,0) {$S_5$};
\node (Y3) at (1,0) {$Y_3$};
\node (S3) at (2,0.5) {$S_3$};
\node (S4) at (2,-0.5) {$S_4$};
\node (Y1)at (3,0.5) {$Y_1$};
\node (Y2)at (3,-0.5) {$Y_2$};
\node (S1)at (4,-0.5) {$S_1$};
\node (S2)at (4,0.5) {$S_2$};

\path[draw] (S5) to (Y3);
\path[draw] (S3) to (Y3);
\path[draw] (S4) to (Y3);
\path[draw] (S3) to (Y1);
\path[draw] (S4) to (Y2);
\path[draw] (Y1) to (Y2);
\path[draw] (S2) to (Y1);
\path[draw] (S1) to (Y2);
\path[draw] (S2) to (S1);
\draw[red] (0.5,0) ellipse (0.9cm and 0.4cm);
\draw[blue] (2.15,0) ellipse (1.4cm and 1cm);
\draw[red] (3.5,0) circle (1cm);
\node[red] at (-0.2,-0.6) {$H_1$}; 
\node[blue] at (0.8,0.9) {$H_3$}; 
\node[red] at (4.5,0.9) {$H_2$}; 
\end{tikzpicture}
\end{center}
\noindent
The non-interacting graphs $H_1,H_2,H_3$ are colored. If total amounts $\oS_1,\oS_2,\oS_3$ are provided then the steady state solutions must satisfy:
$\oS_1=S_5+Y_3$, $\oS_2=S_1+S_2+Y_1+Y_2$, and $\oS_3=S_3+S_4+Y_1+Y_2+Y_3$.
These conserved total amounts are easily  verified by differentiation using \eqref{maindiffeqs}. 

Consider a two-layer cascade of modification cycles that share the demodification enzyme $F$ in each layer. The reaction system consists of $\mmS=\{E,F,S_1,S_2,S_3,S_4\}$, $\mmY=\{Y_1,Y_2,Y_3,Y_4\}$ and the reactions
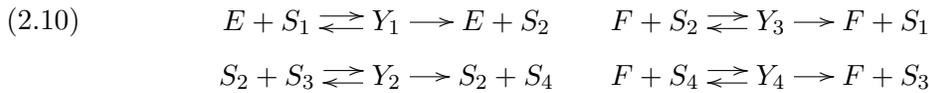
\begin{equation} \label{nopart}
\xymatrix@R=5pt@C=15pt{  E + S_1 \ar@<0.5ex>[r] & Y_1 \ar[r] \ar@<0.5ex>[l] &  E+S_2  &     F + S_2 \ar@<0.5ex>[r]  & Y_3 \ar[r] \ar@<0.5ex>[l] &  F+S_1  \\
  S_2+ S_3 \ar@<0.5ex>[r] & Y_2  \ar@<0.5ex>[l]  \ar[r] &  S_2+S_4 &  F + S_4 \ar@<0.5ex>[r] & Y_4 \ar[r] \ar@<0.5ex>[l] &  F+S_3}   \end{equation}
 The subsets $\mmS_{\alpha}=\{E,S_3,S_4\}, \{E,F\}, \{S_1,S_2\}$ are examples of maximal cuts (they cannot be extended to larger cuts by including more substrates).
The graph $G_{\mmS,\mmY}$ is
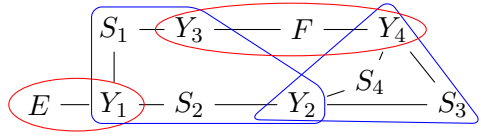
\begin{center}
\begin{tikzpicture}
\node (E) at (0,0) {$E$};
\node (Y1) at (1,0) {$Y_1$};
\node (S2) at (2,0) {$S_2$};
\node (Y2)at (3.5,0) {$Y_2$};
\node (S3)at (5.5,0) {$S_3$};
\node (S4) at (4.4,0.3) {$S_4$};
\node (S1) at (1,1) {$S_1$};
\node (Y3) at (2,1) {$Y_3$};
\node (Y4) at (4.7,1) {$Y_4$};
\node (F) at (3.5,1) {$F$};
\path[draw] (E) to (Y1);
\path[draw] (S2) to (Y1);
\path[draw] (S1) to (Y1);
\path[draw] (S2) to (Y2);
\path[draw] (S3) to (Y2);
\path[draw] (S4) to (Y2);
\path[draw] (S1) to (Y3);
\path[draw] (S4) to (Y4);
\path[draw] (F) to (Y3);
\path[draw] (S3) to (Y4);
\path[draw] (F) to (Y4);
\draw[red] (0.5,0) ellipse (0.9cm and 0.35cm);
\draw[rounded corners,blue] (2.8,-0.2) -- (4.65,1.4) -- (5.9,-0.25) -- cycle;
\draw[rounded corners,blue] (0.7,1.3) -- (2.1,1.3) -- (3.8,0.2) -- (3.8,-0.25) -- (0.7,-0.25) -- cycle;
\draw[red] (3.35,1) ellipse (1.8cm and 0.35cm);
\end{tikzpicture}
\end{center}
These graphs are obtained as connected components of the graph $G_{\mmS_{\alpha},\mmY_{\alpha}}$ for the cuts $\mmS_{\alpha}$ above. As in the previous example, the different non-interacting graphs yield independent conservation laws and thus  if total amounts are provided, we obtain the following equations: $\oS_1 = E+Y_1$,  $\oS_2 = S_3+S_4+Y_2+Y_4$, $\oS_3=F+Y_3+Y_4$ and $\oS_4=S_1+S_2+Y_1+Y_2+Y_3$.

This procedure provides an easy construction of conservation laws. In the two examples above, the conservation laws obtained from the graph are independent and, additionally, determine all  conservation laws arising from $\Gamma^{\perp}$ ($\dim \Gamma^{\perp}=3$ and $4$, respectively).
However, this is not always the case. 
Consider  for instance the reaction system
\begin{equation}\label{weird1}
\xymatrix@C=15pt{  S_1+ S_2 \ar@<0.5ex>[r]  & Y_1 \ar@<0.5ex>[r]  \ar@<0.5ex>[l] & Y_2 \ar@<0.5ex>[r] \ar@<0.5ex>[l] &  S_3+S_4  \ar@<0.5ex>[l] }\end{equation}
The graph $G_{\mmS,\mmY}$ is
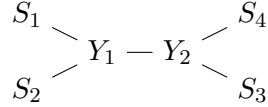
\begin{center}
\begin{tikzpicture}
\node (Y1) at (1,0) {$Y_1$};
\node (S1) at (0,0.5) {$S_1$};
\node (Y2)at (2,0) {$Y_2$};
\node (S2)at (0,-0.5) {$S_2$};
\node (S3) at (3,-0.5) {$S_3$};
\node (S0) at (3,0.5) {$S_4$};
\path[draw] (S1) to (Y1);
\path[draw] (S2) to (Y1);
\path[draw] (S0) to (Y2);
\path[draw] (Y1) to (Y2);
\path[draw] (S3) to (Y2);
\end{tikzpicture}
\end{center}
There are $4$   non-interacting graphs that give the conserved total amounts $\oS_1 =S_1+S_3+Y_1+Y_2$, $\oS_2 =S_1+S_4+Y_1+Y_2$,  $\oS_3 =S_2+S_3+Y_1+Y_2$, and $\oS_4 =S_2+S_4+Y_1+Y_2$. The rank of the space generated by the corresponding $4$ vectors in $\R^{6}$ is $3$, implying that one of the relations is redundant.  
In this case the procedure still gives all conservation laws, because the dimension of $\Gamma^{\perp}$ is $3$.

Consider the following reaction system:
\begin{equation}\label{weird2}
\xymatrix@C=15pt{  S_1+ S_3 \ar@<0.5ex>[r] & Y_1 \ar@<0.5ex>[r] \ar@<0.5ex>[l]  &  S_2+S_4  \ar@<0.5ex>[l]  }\qquad \xymatrix@C=15pt{  S_1+ S_4 \ar@<0.5ex>[r] & Y_2 \ar@<0.5ex>[l]   }\qquad \xymatrix@C=15pt{  S_2+ S_3 \ar@<0.5ex>[r] & Y_3 \ar@<0.5ex>[l]   }
\end{equation}
The graph $G_{\mmS,\mmY}$ is
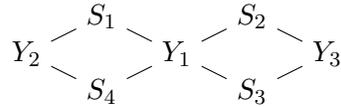
\begin{center}
\begin{tikzpicture}
\node (Y1) at (2,0) {$Y_1$};
\node (S1) at (1,0.5) {$S_1$};
\node (Y2)at (0,0) {$Y_2$};
\node (Y3)at (4,0) {$Y_3$};
\node (S2)at (3,0.5) {$S_2$};
\node (S3) at (3,-0.5) {$S_3$};
\node (S4) at (1,-0.5) {$S_4$};
\path[draw] (S1) to (Y1);
\path[draw] (S2) to (Y1);
\path[draw] (S1) to (Y2);
\path[draw] (Y1) to (S3);
\path[draw] (Y1) to (S4);
\path[draw] (S4) to (Y2);
\path[draw] (S3) to (Y3);
\path[draw] (S2) to (Y3);
\end{tikzpicture}
\end{center}
There are $2$   non-interacting graphs that give the conserved total amounts $\oS_1 =S_1+S_2+Y_1+Y_2+Y_3$, and $\oS_2 =S_3+S_4+Y_1+Y_2+Y_3$. However, $\dim \Gamma^{\perp}=3$ and the procedure fails to provide three independent conservation laws. A third conservation law is $\oS_3=S_1+S_4+Y_1+2Y_2$, and the coefficient $2$ of $Y_2$ cannot be obtained from non-interacting graphs.

 \subsection{Cuts of $\mmS$ and conservation laws}  We provide a criterion to guarantee that there are $\dim \Gamma^{\perp}$ independent conservation laws derived from non-interacting graphs. The criterion will be used in Section \ref{examples}. 

In the following we make use of Lemma~\ref{cut_graph} without further reference. Let $\mmS_{\alpha}$ be a cut with associated set $\mmY_{\alpha}$. Define $\mmS_{\alpha}^c=\mmS\setminus \mmS_{\alpha}$ and $\mmY_{\alpha}^c=\mmY\setminus \mmY_{\alpha}$, and let $N_{\alpha}$, $P_{\alpha}$ (resp. $N_{\alpha}^c$, $P_{\alpha}^c$) be the cardinality of  $\mmS_{\alpha}$, $\mmY_{\alpha}$ (resp. $\mmS_{\alpha}^c$, $\mmY_{\alpha}^c$). 
Extend the set of conservation laws derived from the connected components of $G_{\mmS_{\alpha},\mmY_{\alpha}}$ to a maximal set of $n$ independent conservation laws derived from {\it other} non-interacting graphs (thus containing species in $\mmS_{\alpha}^c\cup \mmY_{\alpha}^c$). Let $n_{\alpha}^c=n-n_{\alpha}$,  where $n_{\alpha}$ is the number of connected components of $G_{\mmS_{\alpha},\mmY_{\alpha}}$.

\begin{lemma} \label{conslaws} Let $\mmS_{\alpha}$ be a cut and keep the notation introduced above.  Then, we have that $\dim \left(  \langle \mmS_{\alpha}^c\cup \mmY_{\alpha}^c \rangle \cap \Gamma  \right) \leq N_{\alpha}^c+P_{\alpha}^c-n_{\alpha}^c$ and   $\dim \Gamma^{\perp}=n$
if and only if 
$$\dim \left(  \langle \mmS_{\alpha}^c\cup \mmY_{\alpha}^c \rangle \cap \Gamma \right)= N_{\alpha}^c+P_{\alpha}^c-n_{\alpha}^c. $$
\end{lemma}
\begin{proof}  Without loss of generality we can assume that $\mmY_{\alpha}=\{Y_1,\dots,Y_{P_{\alpha}}\}$ and $\mmS_{\alpha}=\{S_1,\dots,S_{N_{\alpha}}\}$. Identify $\R^{N+P}$ with $\R^{N_{\alpha}}\times \R^{P_{\alpha}} \times \R^{N_{\alpha}^c}\times \R^{P_{\alpha}^c}$ and let 
\begin{align*}
\Gamma_{\alpha} & = \langle A-B |\ \textrm{for each edge } \xymatrix@C=10pt{A \ar@{-}[r] & B} \textrm{ in }G_{\mmS_{\alpha},\mmY_{\alpha}}\rangle.
\end{align*}
The space $\Gamma_{\alpha}^\perp$ is generated by the vectors which are sums of species in each connected component of  $G_{\mmS_{\alpha},\mmY_{\alpha}}$ and hence $\dim \Gamma_{\alpha}^{\perp}=n_{\alpha}$. 
We have  $\dim \Gamma^\perp\geq n= n_{\alpha}+n_{\alpha}^c$ and we want to determine when  equality holds. Equivalently, we want to see when $\dim \Gamma = N+P-n$. If this is not the case, then  $\dim \Gamma < N+P-n$. Note that $N=N_{\alpha}+N_{\alpha}^c$ and  $P=P_{\alpha}+P_{\alpha}^c$.

Note that $\dim \Gamma_{\alpha}= N_{\alpha}+P_{\alpha} - n_{\alpha}$. Let 
$\pi\colon \R^{N+P} \rightarrow \R^{N_{\alpha}+P_{\alpha}}$ denote the projection onto the first $N_{\alpha}+P_{\alpha}$ coordinates and $\pi_{\alpha}\colon \Gamma \rightarrow \Gamma_{\alpha}$ its restriction to $\Gamma$ ($\pi_{\alpha}$ a surjective map). Then, 
$\dim \Gamma = \dim \Gamma_{\alpha} + \dim \ker \pi_{\alpha}$ and so $ \dim \ker \pi_{\alpha} \leq N_{\alpha}^c+P_{\alpha}^c-n_{\alpha}^c$. Further, $\dim \Gamma^\perp= n$  if and only if $ \dim \ker \pi_{\alpha} = N_{\alpha}^c+P_{\alpha}^c-n_{\alpha}^c$.
Finally, note that  $ \langle \mmS_{\alpha}^c\cup \mmY_{\alpha}^c \rangle \cap \Gamma = \ker\pi_{\alpha}$. Indeed, let 
$i\colon \Gamma  \hookrightarrow  \R^{N+P}$ and $i_{\alpha}\colon \Gamma_{\alpha}  \hookrightarrow  \R^{N_{\alpha}+P_{\alpha}}$ denote the natural inclusions. We have that $i_{\alpha} \circ \pi_{\alpha}= \pi \circ i$. The kernel of $\pi$ is clearly $\R^{N_{\alpha}^c+P_{\alpha}^c}=\langle \mmS_{\alpha}^c\cup \mmY_{\alpha}^c \rangle$ from where it follows that  the kernel of $\pi_{\alpha}$ is $\langle \mmS_{\alpha}^c\cup \mmY_{\alpha}^c \rangle\cap \Gamma$.

Therefore,  $\dim \left(  \langle \mmS_{\alpha}^c\cup \mmY_{\alpha}^c \rangle \cap \Gamma \right)=\dim \ker \pi_{\alpha} = N_{\alpha}^c+P_{\alpha}^c-n_{\alpha}^c$ if and only if $\dim \Gamma^{\perp}=n$ and the lemma is proved.
\end{proof}

As each non-interacting graph corresponds to a minimal cut, the lemma above provides a condition for when all conservation laws are recovered from cuts.

\emph{Remark.}
An easy way to construct elements of $\langle \mmS_{\alpha}^c\cup \mmY_{\alpha}^c  \rangle \cap \Gamma$ is by considering:
\begin{enumerate}[(i)]
\item Vectors $S_i-S_j$ for any pair $S_i,S_j \in \mmS_{\alpha}^c$ for which there exists a chain of reactions $\xymatrix@C=10pt{S_m+ S_i \ar@{-}[r] & A_{1} \ar@{-}[r] & \dots \ar@{-}[r] & A_{r} \ar@{-}[r] & S_m+S_j  }$
for some $S_m$, $A_{u}\in \mathcal{C}$ and  $-$ is  $\leftarrow$ or $\rightarrow$.  
\item Vectors $S_i+S_j-Y_k$, $S_i-S_j$ or $Y_k-Y_v$ corresponding to reactions with $S_i,S_j\in \mmS_{\alpha}^c$ and $Y_k,Y_v\in \mmY_{\alpha}^c$. 
\end{enumerate}
If we can construct $N_{\alpha}^c+P_{\alpha}^c-n_{\alpha}^c$ independent elements of $\langle \mmS_{\alpha}^c\cup \mmY_{\alpha}^c  \rangle \cap \Gamma$  of the previous type, then  the previous lemma holds. 
 
In Example \eqref{mainex} consider the cut $\mmS_{\alpha}=\{S_1,S_2,S_5\}$ with $\mmS_{\alpha}^c=\{S_3,S_4\}$ and the given conservation laws ($n=3$). We have $N_{\alpha}^c=2$ and $n_{\alpha}^c=1$. Further, $\mmY_{\alpha}=\mmY$ so that $P_{\alpha}^c=0$. The element $ S_3-S_4= (S_3+S_5 - Y_3)- (S_4+S_5-Y_3)$
belongs to $\langle \mmS_{\alpha}^c\rangle \cap \Gamma$. In addition, $N_{\alpha}^c+P_{\alpha}^c-n_{\alpha}^c=1$ and thus 
$\dim (\langle \mmS_{\alpha}^c\rangle \cap \Gamma)=N_{\alpha}^c+P_{\alpha}^c-n_{\alpha}^c$, implying  that all conservation laws are found from non-interacting graphs.

In Example \eqref{nopart}, consider the cut $\mmS_{\alpha}= \{E,S_3,S_4\}$ with $\mmS_{\alpha}^c =\{F,S_1,S_2\}$ and $N_{\alpha}^c=3$. In this case, $\mmY_{\alpha}=\{Y_1,Y_2,Y_4\}$, $\mmY_{\alpha}^c=\{Y_3\}$ and so $P_{\alpha}^c=1$. 
Two of the four conservation laws involve elements   in $\mmS_{\alpha}\cup \mmY_{\alpha}$ only and hence $n_{\alpha}=n_{\alpha}^c=2$. Further,  $N_{\alpha}^c+P_{\alpha}^c-n_{\alpha}^c=2$. The two independent vectors  $F+S_2-Y_3$ and $F+S_1-Y_3$ belong to $\langle F,S_1,S_2,Y_3 \rangle \cap \Gamma$. Thus, the graphical procedure provides all conservation laws.

In Example \eqref{weird1}, consider the cut $\mmS_{\alpha}=\{S_1,S_3\}$ with $\mmS_{\alpha}^c=\{S_2,S_4\}$ so that $N_{\alpha}^c=2$, $P_{\alpha}^c=0$. There is only one conservation law in $\mmS_{\alpha}\cup \mmY_{\alpha}$, $S_3+S_1+Y_1+Y_2$, and since $n=3$, then $n_{\alpha}^c=2$.  It follows that $N_{\alpha}^c+P_{\alpha}^c-n_{\alpha}^c=0$, and we are guaranteed that the dimension of $\langle S_2,S_4\rangle \cap \Gamma$ is zero.

In Example \eqref{weird2}, consider the cut $\mmS_{\alpha}=\{S_1,S_2\}$ with $\mmS_{\alpha}^c=\{S_3,S_4\}$ and  $N_{\alpha}^c=2$, $P_{\alpha}^c=0$. We have $n_{\alpha}^c=1$ and $N_{\alpha}^c+P_{\alpha}^c-n_{\alpha}^c=1$. However, $\langle S_3,S_4 \rangle \cap \Gamma$ has dimension zero and thus not all conservation laws arise from non-interacting graphs.

\section{Variable elimination}
In this section we show that the intermediate complexes  can always be eliminated and expressed as polynomials in the substrates with coefficients in $\R(\Con)$ (Section \ref{eliminter}). After choosing a cut $\mmS_{\alpha}$, the substrates in $\mmS_{\alpha}$ can be expressed in terms of those in $\mmS_{\alpha}^c=\mmS\setminus\mmS_{\alpha}$ (Section \ref{partition}).

\subsection{Elimination of intermediate complexes}\label{eliminter}
Consider the system $\dot{Y}_i=0$ in \eqref{dY1} as a linear system of $P$ polynomial equations with coefficients in $\R[\Con\cup\, \mathcal{S}]$ and $P$ variables $Y_1,\dots,Y_P$. 
If  the system has maximal rank, then there is a unique solution in $\R(\Con\cup\, \mathcal{S})$.

Specifically, we have a linear system  $AY=z$ where $Y=(Y_1,\dots,Y_P)^{t}$ and 
 $A=\{\lambda_{k,v}\}$ is a $P\times P$ matrix with coefficients in $\R[\Con]$,
 $$\lambda_{k,v} =\begin{cases}   -c_{v,k} &   \textrm{if } k\neq v\\
 \sum_{j=1}^N \sum_{i=1}^j b_{i,j}^k  + \sum_{u=1}^P c_{k,u} &  
\textrm{if } k=v.
\end{cases}$$  The independent term  $z=(z_1,\dots,z_P)^t$ is in $\R[\Con\cup\, \mathcal{S}]$:
$z_k =  \sum_{i\leq j} a_{i,j}^k S_{i}S_j.  $

Assume that $A$ has maximal rank $P$ in $\R(\Con)$. Then, using Cramer's rule to solve linear systems of equations, 
we obtain that $Y_k = \rho_k/\rho$ with $\rho=\det(A)\neq 0$ and $\rho_k$ the determinant of $A$ with the $k$-th column substituted by $z$.  Since the determinant is a homogeneous polynomial in the entries of the matrix, it follows that $\rho \in \R[\Con]$ and $\rho_k\in \R[\Con\cup\, \mathcal{S}]$. Therefore,
$$ Y_k= \sum_{i\leq j} \mu_{i,j}^k S_i S_j $$ 
with $\mu_{i,j}^k\in \R(\Con)$ and thus $Y_k$ is a polynomial in $\R(\Con)[\mmS]$.  If both $\rho,\rho^k$ are $S$-positive elements of  $\R[\Con]$ and $\R[\Con\cup\, \mathcal{S}]$, respectively, then for  positive rate constants and non-negative values of $S_i$, the steady state value of $Y_k$ is non-negative as well.
S-positivity of $\rho,\rho^k$ is proven in the next section using the \emph{Matrix-Tree theorem} \cite{Tutte-matrixtree}. Some basic concepts from graph theory are required.

\medskip
{\bf Graphs and the Matrix-Tree theorem.} Given a directed graph $G$, a \emph{spanning tree} $\tau$ is a directed subgraph with the same node set as $G$ and such that the corresponding undirected graph is connected and acyclic. 
There is a unique undirected path between any two nodes in a spanning tree \cite{Diestel}. A spanning tree $\tau$ is said to be \emph{rooted} at a node $v$ if  the unique path between any node $w$ and $v$ is directed from $w$ to $v$. It follows that $v$ is the only node with no out-edges, that is, there is no edge of the form $v\rightarrow w$ in $\tau$. In addition, there cannot be a node with two out-edges in $\tau$.  The graph $G$ is  \emph{strongly connected} if for any pair of nodes $v,w$ there is a directed path from $v$ to $w$. Any directed path from $v$ to $w$ in a strongly connected graph can be extended to a spanning tree rooted at $w$. Some general references for graph theory are \cite{Diestel} and \cite{Gross-Yellen}.

If $G$  is  labeled, then $\tau$ inherits a labeling from $G$ and we define
$$\pi(\tau)=\prod_{x\xrightarrow{a}y \in \tau} a.$$
Assume that $G$ has no self-loops. Order the node set $\{v_1,\dots,v_n\}$ of $G$  and denote by $a_{i,j}$ the label of the edge $v_i\rightarrow v_j$. We set $a_{i,j}=0$ if there is no edge from $v_i$ to $v_j$ (thus $a_{i,i}=0$). Let $\mathcal{L}(G)=\{\alpha_{i,j}\}$ be the \emph{Laplacian} of $G$, that is the matrix with
$$\alpha_{i,j} = \begin{cases} 
a_{j,i} & \textrm{if } i\neq j \\
-\sum_{k=1}^n a_{i,k} & \textrm{if }i=j,
\end{cases} $$
such that the column sums are zero. For each node $v_{j}$, let $\Theta(v_j)$ be the set of spanning trees of $G$ rooted at $v_{j}$. Then, the Matrix-Tree theorem states that the maximal minor $\mathcal{L}(G)_{(ij)}$  (the determinant of the minor obtained by removing the $i$-th row and  the $j$-th column of $\mathcal{L}(G)$) is:
$$ \mathcal{L}(G)_{(ij)} = (-1)^{n-1+i+j}  \sum_{\tau \in \Theta(v_j)}  \pi(\tau).$$ 
Note that for notational simplicity we have defined the Laplacian as the transpose of how it is   usually defined and the Matrix-Tree theorem has been adapted consequently.

In our case, the matrix   $A$ is not a Laplacian, since the column sums $\sum_{j=1}^N \sum_{i=1}^j b_{i,j}^k$ are not zero. However,   $A$ can be extended such that its determinant is a maximal minor of a Laplacian.

\subsection{Decomposition of the system}
Let $G_{\mathcal{Y}}$ be the directed graph  with node set $\mathcal{Y}$ and a directed edge   $Y_k\rightarrow Y_v$ if $(k,v)\in I_c$. The node sets of the connected components of $G_{\mathcal{Y}}$ determine a partition of
 $\mathcal{Y}$:   $\mathcal{Y}=\mathcal{Y}_1\cup \dots \cup \mathcal{Y}_s$.  
 Let $P_l$ be the cardinality of $\mathcal{Y}_l$ and  rename the intermediate complexes  such that  $\mathcal{Y}_{l}=\{Y_{P_1+\dots+P_{l-1}+1},\dots,Y_{P_1+\dots+P_{l}}\}$. 

If $Y_k\in \mathcal{Y}_l$ for some $l$, then $c_{k,v}=c_{v,k}=0$ for any $v$ such that $Y_v\notin \mathcal{Y}_l$. It follows, that $A$ is a block diagonal matrix $\diag(A_1,\ldots,A_s)$ with $A_l$ being a  $P_l\times P_l$ matrix. Solving $AY=z$ is thus equivalent to solving $s$ ``smaller'' systems with matrices $A_l$. Further,  $A$ has maximal rank $P$ if and only if $A_l$ has maximal rank $P_l$ for all $l$. 

Consider the connected component  $G_{\mathcal{Y}_l}$ corresponding to $\mathcal{Y}_l$. 
We construct an extended labeled directed graph $\widehat{G}_{\mathcal{Y}_l}$ with node set  $\mathcal{Y}_l\cup \{*\}$.  For convenience we order the nodes such that  $Y_{P_1+\dots+P_{l-1}+k}$ is  the $k$-th node and
 $*$ the $(P_l+1)$-th  node. Let $b^k=\sum_{i\leq j}  b_{i,j}^k$ and 
$a^k= \sum_{i\leq j} a_{i,j}^k$. The graph $\widehat{G}_{\mathcal{Y}_l}$ has the following labeled directed edges:
$Y_k\xrightarrow{c_{k,v}} Y_v$ if $(k,v)\in I_c$,  $Y_k\xrightarrow{b^k}  *$ if $b^k\neq 0$, and 
  $*\xrightarrow{a^k} Y_k$ if $a^k\neq 0$.

In Example \eqref{mainex},  the graph $G_{\mmY}$ has two connected components $\mmY_1=\xymatrix{Y_1\ar@<0.3ex>[r] & Y_2\ar@<0.3ex>[l]}$ and $\mmY_2=Y_3$. The graphs $\widehat{G}_{\mmY_1}$ and $\widehat{G}_{\mmY_2}$ are

$$
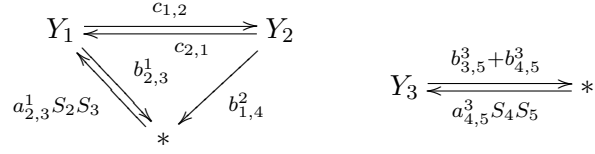
\xymatrix@R=6pt{ Y_1 \ar@<0.3ex>[ddr]^(0.6){b_{2,3}^1} \ar@<0.3ex>[rr]^{c_{1,2}} && Y_2\ar@<0.3ex>[ll]^(0.4){c_{2,1}} \ar[ddl]^{b_{1,4}^2} \\ &   & \qquad &   Y_3 \ar@<0.3ex>[rr]^{b_{3,5}^3 + b_{4,5}^3} &&  \ast \ar@<0.3ex>[ll]^{a_{4,5}^3S_4S_5}  
\\ & \ast \ar@<0.3ex>[uul]^{a_{2,3}^1S_2S_3}  }$$ 

Let $\mathcal{L}=\{\alpha_{k,v}\}$ be the Laplacian of $\widehat{G}_{\mathcal{Y}_l}$.  If $k,v\leq P_l$, then $\alpha_{k,v}=-\lambda_{k,v}$. The entries of the last row  are $\alpha_{P_l+1,k} = b^k$ for $k\leq P_l$ and the entries of the last  column are $\alpha_{k,P_l+1}= a^k$ ($=z_k$) for $k\leq P_l$. 
We conclude that the $(P_l+1,P_l+1)$ principal minor of $\mathcal{L}$ is exactly $-A_l$ and thus,  by the Matrix-Tree theorem, we have
$$\det(A_l) = (-1)^{P_l} \mathcal{L}_{(P_l+1,P_l+1)} =   \sum_{\tau \in \Theta(*)}  \pi(\tau).  $$ 
Assumption (ii) of a PTM system ensures that each $Y_k$ ultimately reacts to some $S_i+S_j $ via $\mmY$, and hence there exists at least one spanning tree rooted at $*$. Thus, $\det(A_l)\neq 0$ and  $\det(A_l)$ is an S-positive element of $\R[\Con]$.

By the definition of $\rho_k$ and  the Matrix-Tree theorem,
$$\rho_k=   (-1)^{k+1}  \mathcal{L}_{(P_l+1,k)} =   \sum_{\tau \in \Theta(Y_{P_{l-1}+k})}  \pi(\tau),  $$ 
and hence $\rho_k$ is either zero or  an S-positive element of $\R[\Con \cup\, \mathcal{S}]$. 
 
 If there exists at least one spanning tree rooted at $v_k=Y_{P_{l-1}+k}$, then $\rho_k\neq 0$. 
 A necessary condition for this to happen is the existence of at least one in-edge to $v_k$. Otherwise the concentration at steady state of $v_k$ is zero, which is  expected if $v_k$ is only consumed and never produced. Similarly, if there is no reaction of the form $S_i+S_j\rightarrow Y_{P_{l-1}+m}$ for any $m$ (that is, a directed edge $*\rightarrow v_m$), then  $\rho_k=0$ for all $k$. 

The term $\rho_k$ is a homogeneous polynomial of degree 2 in $\mathcal{S}$ with coefficients in $\R[\Con]$, because any  spanning tree rooted at a node $v_k$ has exactly one edge of the form $*\rightarrow v_m$ for some $m$. Further, a monomial $S_iS_j$ appears in $\rho_k$ only if $S_i+S_j$ ultimately reacts to $v_k$ via $\mmY_l$.
If $\widehat{G}_{\mathcal{Y}_l}$ is strongly connected, then this condition is both sufficient and necessary. Indeed, if $S_i+S_j$ ultimately reacts to $v_k$ via $\mmY_l$, then there is a spanning tree rooted at $v_k$ containing this path.

The next proposition summarizes the discussion above:

\begin{proposition}\label{Yelim}
Consider a PTM system with  intermediate complexes $\mathcal{Y}$ and  substrates $\mathcal{S}$.  Then, $\dot{Y}_k=0$ for all $k$, if and only if
\begin{equation}\label{yelim} Y_k= \sum_{i\leq j} \mu_{i,j}^k S_i S_j \end{equation}
with $\mu_{i,j}^k\in \R(\Con)$ being  either zero or S-positive. Further:
\begin{enumerate}[(i)]
\item If  $S_i+S_j$ does not ultimately react to $Y_k$ via $\mmY$, then $\mu_{i,j}^k=0$. 
\item  If $\widehat{G}_{\mathcal{Y}_l}$ is strongly connected and $Y_k\in \mmY_l$, then $\mu_{i,j}^k\neq 0$   if and only if $S_i+S_j$ ultimately reacts to $Y_k$ via $\mmY_l$.
\item $\widehat{G}_{\mathcal{Y}_l}$ is strongly connected if and only if in \eqref{yelim},  $Y_k$ is a non-zero polynomial in $\R(\Con)[\mmS]$ for all $Y_k\in{\mmY}_l$.
\end{enumerate}
\end{proposition}

\emph{Remark.}
The condition that $\widehat{G}_{\mathcal{Y}_l}$ is strongly connected is biochemically reasonable: The intermediate complexes are not the initial or final products of the system and should eventually be broken up into parts.

In Example \eqref{mainex}, 
the graph $\widehat{G}_{\mmY_1}$ has three spanning trees rooted at $*$ 
so that $\det(A_1)= b_{1,4}^2c_{1,2} + b_{2,3}^1c_{2,1}+  b_{1,4}^2b_{2,3}^1$. 
There is one spanning tree rooted at $Y_2$, giving $\rho_2=c_{1,2}a_{2,3}^1S_2S_3$, and two spanning trees rooted at $Y_1$, giving $\rho_1=(b_{1,4}^2+c_{2,1}) a_{2,3}^1S_2S_3$.  The graph $\widehat{G}_{\mmY_2}$ has one spanning tree rooted at $*$ so that $\det(A_2)= b_{3,5}^3+b_{4,5}^3$, and one spanning tree rooted at $Y_3$, giving $\rho_3= a_{4,5}^3S_4S_5$. Thus:
$$ Y_1=\mu_{2,3}^1S_2S_3, \quad Y_2=\mu_{2,3}^2S_2S_3, \quad Y_3= \mu_{4,5}^3  S_4 S_5 $$ 
with $\mu_{2,3}^1=\frac{(b_{1,4}^2+c_{2,1}) a_{2,3}^1}{\det(A_1)}$, $\mu_{2,3}^2=\frac{c_{1,2}a_{2,3}^1}{\det(A_1)}$, and $\mu_{4,5}^3 = \frac{a_{4,5}^3}{\det(A_2)}. $

\begin{lemma}\label{same}
Let $\widehat{G}_{\mmY}=\cup_{l} \widehat{G}_{\mmY_l}$. The graphs $\widehat{G}_{\mmY_l}$, $l=1,\ldots,s$, are strongly connected if and only if the graph $\widehat{G}_{\mmY}$ is. 
\end{lemma}
\begin{proof}
Assume that  the graphs $\widehat{G}_{\mmY_l}$ are strongly connected. Then, for any $v\in\widehat{G}_{\mmY_l}$ and $\omega \in  \widehat{G}_{\mmY_j}$, there are  directed paths  
 $v \rightarrow \ast$ in $\widehat{G}_{\mmY_l}$ and  $\ast \rightarrow v$ in $\widehat{G}_{\mmY_j}$, which by composition give a directed path between $v$ and $w$.

For the  reverse implication, let $v,\omega$ be two elements of $\mmY_l$. Since  $\widehat{G}_{\mmY}$ is strongly connected, there exists a directed path $\alpha\colon v\rightarrow w$ in $\widehat{G}_{\mmY}$. We can assume that  $v,w\neq \ast$. A path connecting an intermediate complex in $\mmY_l$ to one in $\mmY_j$ for $j\neq l$ must pass through $\ast$. If a path $\alpha$ goes through $\widetilde{v}\in \mmY_j$, for $j\neq l$, then it must go through $\ast$, first in and then out, potentially many times until it goes back to $\mmY_l$ and to $w$. Therefore, $\alpha$ has the form $v\xrightarrow{\alpha_{1}}\ast\xrightarrow{\beta} \ast\xrightarrow{\alpha_{2}} w$ with $\alpha_{1}$ and $\alpha_{2}$ being paths in $\widehat{G}_{\mmY_l}$. It follows that the path $v\xrightarrow{\alpha_{1}} \ast\xrightarrow{\alpha_{2}} w$ is a directed path from $v$ to $w$ in $\widehat{G}_{\mmY_l}$. 
\end{proof}

\subsection{Elimination of substrates}\label{partition}
Equation~\eqref{yelim}
shows that at  steady state the intermediate complexes are given as zero or S-positive rational functions in the substrates and the rate constants.  Insertion of \eqref{yelim}
into  the (time dependent) differential equations for the substrates   is the procedure known as the \emph{quasi-steady state} assumption. The rationale is that intermediate complexes tend to reach steady state much faster than substrates and thus some variables in the dynamical system can be eliminated. We have shown here that PTM systems ``mathematically'' enable  this simplification although justification is required in concrete examples.

We  now use the steady state equation \eqref{dS1}  to further eliminate some of the substrates in terms of others.  Recall equation  \eqref{dS1}, that  is $\dot{S}_i=0$,
\begin{align}\label{dS2-0}
0 =& \sum_{j=1}^N \sum_{k=1}^P \epsilon_{i,j} (-  a_{i,j}^k S_iS_j +b_{i,j}^k Y_{k})+  \sum_{j=1}^N (d_{j,i}S_j - d_{i,j}S_i) 
\end{align}
for  $i=1,\dots,N$.  After substitution of the values for  $Y_k$, we have 
\begin{equation}\label{dS2}
0=\sum_{u=1}^N \sum_{k=1}^P \sum_{j\leq t} \epsilon_{i,j}  b_{i,u}^k\mu_{j,t}^k S_j S_t   -\sum_{j=1}^N \sum_{k=1}^P \epsilon_{i,j}  a_{i,j}^k S_iS_j +  \sum_{j=1}^N (d_{j,i}S_j - d_{i,j}S_i),
\end{equation}
These equations are quadratic in  $\mathcal{S}$. To proceed with linear elimination  it is necessary to decide which variables are to be eliminated and which will be taken as part of the coefficient field. Since a monomial $S_iS_j$ appears only if $S_i$ and $S_j$ interact, we can proceed as long as $\mmS$ can be partitioned in an appropriate way. 

\begin{lemma}\label{Slinear}
Assume that there is a cut   $\mmS_{\alpha}$ with associated set $\mmY_{\alpha}$ (Definition~\ref{def1}).  Then,  \eqref{dS2}  for the substrates in $\mmS_{\alpha}$
 is a homogeneous  linear system of equations  in the substrates  $\mmS_{\alpha}$ and with coefficients in $\R(\Con\cup\,\mmS_{\alpha}^c)$.
\end{lemma}
\begin{proof}
For the three sums in \eqref{dS2} we make the following observations: If $S_i\in\mmS_{\alpha}$ and  $d_{j,i}\neq 0$, then also $S_j\in\mmS_{\alpha}$.  If $S_i\in\mmS_{\alpha}$ and  $a_{i,j}^k\neq 0$, then $S_j\not\in\mmS_{\alpha}$, otherwise $S_i$ and $S_j$ would  interact.
Finally, if $\mu_{j,t}^k\neq 0$, then according to Proposition~\ref{Yelim},  $S_j+S_t$ ultimately reacts to $Y_k$ via $\mmY_{\alpha}$.  Hence, since $\mmS_{\alpha}$ is a cut, one of  $S_j$ and $S_t$ (but not both) belongs to $\mmS_{\alpha}$.  Thus  \eqref{dS2} for the substrates in $\mmS_{\alpha}$ is a homogeneous linear system of equations in the species in $\mmS_{\alpha}$.  
\end{proof}

Assume that there exists a cut $\mmS_{\alpha}$ and that $\mmS_{\alpha}=\{S_1,\dots,S_{N_{\alpha}}\}$. 
 It follows that for $S_i\in \mmS_{\alpha}$ the equations in \eqref{dS2} form an $N_{\alpha}\times N_{\alpha}$ homogeneous linear system of equations with variables $\mmS_{\alpha}$ and coefficients in $\R(\Con\cup\, \mmS_{\alpha}^c)$. Further, $\epsilon_{i,j}=1$ if $S_i\in \mmS_{\alpha}$. 
 
 Let $B$ be the matrix with entries $\widetilde{b}_{i,j}$ for $i\not=j$ and $\widetilde{b}_{i,i} -  \widetilde{a}_i$ for $i=j$, where
$$ \widetilde{a}_i  = \sum_{j=1}^{N_1}  d_{i,j} + \sum_{j=N_1+1}^N \sum_{k=1}^P a_{i,j}^k S_j, \quad 
\widetilde{b}_{i,j} = d_{j,i}+\sum_{t=N_1+1}^N \sum_{k=1}^P b_i^k\mu_{j,t}^k S_t,
\quad b_i^k =   \sum_{u=N_1+1}^N b_{i,u}^k,$$
so that \eqref{dS2} becomes
\begin{equation}\label{B-system}
0= \sum_{j=1,j\neq i}^{N_1} \widetilde{b}_{i,j}S_j  + (\widetilde{b}_{i,i} -  \widetilde{a}_i) S_i.
\end{equation}

Consider  Example  \eqref{mainex} and the cut $\mmS_{\alpha}=\{S_1,S_2,S_5\}$. Then the equations \eqref{dS2}  are
$ 0 = -d_{1,2}S_1 + b_{1,4}^2\mu_{2,3}^2S_2S_3$ and $0 =d_{1,2}S_1 -a_{2,3}^1 S_2S_3 + b_{2,3}^1 \mu_{2,3}^1S_2S_3$, corresponding to $\dot{S}_1=0$ and $\dot{S}_2=0$, respectively. The equation $\dot{S_5}=0$ is trivial because of the conservation law $\dot{Y}_3+\dot{S}_5=0$. Further, we have
$\widetilde{a}_1= d_{1,2}$, $\widetilde{a}_2= a_{2,3}^1S_3$, $\widetilde{b}_{1,2}=b_{1,4}^2\mu_{2,3}^2S_3$, $  \widetilde{b}_{2,1}=d_{1,2}$, $\widetilde{b}_{2,2}=b_{2,3}^1 \mu_{2,3}^1S_3,$
while  the rest of the coefficients  are zero.

Lemma~\ref{cut_graph} ensures that there is a conservation law for each connected component of  $G_{\mmS_{\alpha},\mmY_{\alpha}}$.  Let $C_1^{\alpha},\dots,C^{\alpha}_{n_{\alpha}}$ be the node sets of  the connected components and define $\mmS_{{\alpha},l}=\mmS_{\alpha} \cap C_l^{\alpha}$ and $\mmY_{{\alpha},l}=\mmY_{\alpha} \cap C_l^{\alpha}$ so that $ \sum_{S_i\in \mmS_{{\alpha},l}} \dot{S}_i +\sum_{Y_k\in  \mmY_{{\alpha},l}}\dot{Y}_k=0$ for  $l=1,\dots,n_{\alpha}$, are conservation laws.  Imposing only that the intermediate complexes are at steady state, that is $\dot{Y_k}=0$ for all $k$, we obtain
\begin{equation}\label{dS0} \sum_{S_i\in \mmS_{{\alpha},l}} \dot{S_i}=0,\qquad l=1,\dots,n_{\alpha}.  
\end{equation}
It follows that the column sums of the matrix $B$ restricted to the rows corresponding to the substrates in $\mmS_{{\alpha},l}$ are all zero.  Consequently,  the matrix $B$ has rank at most  $N_{\alpha}-n_{\alpha}$.

 Let $G_{\mmY_{\alpha}}$  be $G_{\mmY}$  restricted to the nodes $\mmY_{\alpha}$. It follows from the definition of $\mmY_{\alpha}$ (Definition~\ref{def1}) that $G_{\mmY_{\alpha}}$ is a union of connected components of $G_{\mmY}$. Define  $\widehat{G}_{\mmY_{\alpha}}$ similarly (cf. Lemma \ref{same}). Let $N_{\alpha,l}$ be the cardinality of $S_{\alpha,l}$.
 
\begin{lemma}\label{block}
After reordering  the substrates  in $\mmS_{\alpha}$, $B$ is a block diagonal matrix, namely $\diag(B_1,\ldots,B_{n_{\alpha}})$, where $B_l$ is an $N_{\alpha,l}\times N_{\alpha,l}$ matrix.
Further,  if $\widetilde{b}_{i,j}\neq 0$ then there is a reaction $S_j\rightarrow S_i$ or there exist $S_u,S_t\in \mmS_{\alpha}^c$, so that   $S_j+S_t$ ultimately reacts to $S_i+S_u$ via $\mmY_{\alpha}$.
If in addition $\widehat{G}_{\mmY_{\alpha}}$ is strongly connected, then the reverse is true.
\end{lemma}
\begin{proof}
 It follows from Lemma \ref{same} and Proposition \ref{Yelim}(i) that if $\mu_{j,t}^k\neq 0$ then $S_j+S_t$ ultimately reacts to $Y_k$. By definition, $b_i^k\neq 0$ if and only if there exists a reaction $Y_k\rightarrow S_i+S_u$ for some $S_u\in \mmS_{\alpha}^c$. 
We have  $\widetilde{b}_{i,j}\neq 0$ if and only if  $d_{j,i}\neq 0$ or $b_i^k\mu_{j,t}^k\neq 0$ for some $k$ and $t$, 
and hence either there is a reaction $S_j\rightarrow S_i$ or there exist $S_u,S_t\in \mmS_{\alpha}^c$, so that   
$S_j+S_t$ ultimately reacts to $S_i+S_u$ via $\mmY_{\alpha}$. If $\widehat{G}_{\mmY_{\alpha}}$ is strongly connected then by Proposition \ref{Yelim}(ii) the existence of these reactions is a sufficient condition. It follows, after reordering of the  species in $\mmS_{\alpha}$, that $B$ is a block diagonal matrix with blocks given by the species in each connected component of $G_{\mmS_{\alpha},\mmY_{\alpha}}$. Indeed, if $S_i,S_j$ are in different components, then $\widetilde{b}_{i,j}=\widetilde{b}_{j,i}=0$. 
\end{proof}

It follows from the lemma that a necessary condition for $\widetilde{b}_{i,j}\neq 0$ is that $S_i$ can be ``produced'' from $S_j$. We restrict the study to the case where $G_{\mmS_{\alpha},\mmY_{\alpha}}$ is connected and note that the results apply to every connected component individually. However, the propositions to be derived below are stated in full generality, that is, without the assumption that $G_{\mmS_{\alpha},\mmY_{\alpha}}$ is connected.

Using \eqref{dS0}, the column sums of  $B$ are zero. Thus, $B$ is the Laplacian of a labeled directed graph $G_{\mmS_{\alpha}}$ with node set $\mmS_{\alpha}$ and an edge from $S_j$ to $S_i$ whenever  
$\widetilde{b}_{i,j}\neq 0$, $i\not=j$. Note that $\widetilde{b}_{i,j}\in \R(\Con)[\mmS_{\alpha}^c]$ is S-positive.

Since $G_{\mmS_{\alpha},\mmY_{\alpha}}$ is connected, then so is  $G_{\mmS_{\alpha}}$.  In general, two species $S_i,S_j$ belong to the same connected component of $G_{\mmS_{\alpha},\mmY_{\alpha}}$ if and only if they belong to the same connected component of $G_{\mmS_{\alpha}}$. We will use this fact repeatedly in what follows. 

By the Matrix-Tree theorem, the principal minors $B_{(i,j)}$ of $B=\mathcal{L}(G_{\mmS_{\alpha}})$ are 
$$B_{(i,j)} =(-1)^{N_{\alpha}-1+i+j}   \sum_{\tau \in \Theta(S_j)}  \pi(\tau).   $$ 
Thus, $B$ has rank $N_{\alpha}-1$ if and only if there exists at least one spanning tree in $G_{\mmS_{\alpha}}$ rooted at some $S_j$ with $j\in \{1,\dots,N_{\alpha}\}$.  For a general PTM system with a selected cut $\mmS_{\alpha}$, we obtain the following proposition:
\begin{proposition}\label{maxrank}
The non-interacting graphs provide all  conser\-vation laws involving only the substrates $\mmS_{\alpha,l}$ if and only if $G_{\mmS_{\alpha,l}}$  has at least one rooted spanning tree for all $l$. 
\end{proposition}
\begin{proof}
 The non-interacting graphs provide all  conservation laws involving only $\mmS_{\alpha,l}$ if and only if all conservation  laws are  multiples of  $ \sum_{S_i\in \mmS_{{\alpha},l}} S_i +\sum_{Y_k\in  \mmY_{{\alpha},l}}Y_k=0$, which is the case if and only if the rank of $B_l$ is $N_{\alpha,l}-1$. As stated above this is equivalent to the existence of a rooted spanning tree in $G_{\mmS_{\alpha,l}}$.\end{proof}

\emph{Remark. } In particular, the lemma holds  if $G_{\mmS_{\alpha}}$ is strongly connected.
If $\widehat{G}_{\mmY_{\alpha}}$ is strongly connected, then to check that  $G_{\mmS_{\alpha}}$ is strongly connected we do not need to calculate the labels of  $G_{\mmS_{\alpha}}$. Whether  there is an edge or not between two nodes follows from the set of reactions, cf. Lemma~\ref{block}.

For simplicity we assume that there exists a spanning tree rooted at $S_1$. 
Then, the variables $S_2,\dots,S_{N_{\alpha}}$ can be solved in the coefficient field $\R(\Con\cup\, \mmS_{\alpha}^c\cup \{S_1\})$. In particular,  using Cramer's rule  and the Matrix-Tree theorem, we obtain
\begin{equation}\label{Srational}
 S_j = \frac{(-1)^{j+1}B_{(1,j)}}{B_{(1,1)}} = \frac{\sigma_j(\mmS_{\alpha}^c)}{\sigma(\mmS_{\alpha}^c)} S_1=r^S_j(\mmS_{\alpha}^c)S_1, \textrm{ where }
\begin{cases}
\sigma(\mmS_{\alpha}^c) &\hspace{-0.2cm}=  \sum_{\tau \in \Theta(S_1)}  \pi(\tau)\neq 0\\ \sigma_j(\mmS_{\alpha}^c) &\hspace{-0.2cm}=  \sum_{\tau \in \Theta(S_j)}  \pi(\tau)
\end{cases}
\end{equation} 
and $j=2,\dots, N_{\alpha}$. It follows that  $\sigma(\mmS_{\alpha}^c)$ is S-positive and $\sigma_j(\mmS_{\alpha}^c)$ is either a zero or S-positive  element of $\R(\Con)[\mmS_{\alpha}^c]$. 
If the graph $G_{\mmS_{\alpha}}$  is strongly connected, then $\sigma_j(\mmS_{\alpha}^c) \neq 0$ for all $j$ and  any choice of $S_j$ could be used instead of $S_1$. Further:

\begin{proposition}\label{snonzero}
A connected component $G_{\mmS_{\alpha,l}}$ of the graph $G_{\mmS_{\alpha}}$  is strongly connected if and only if $\sigma_j(\mmS_{\alpha}^c)$ is a non-zero rational function in $\R(\Con\cup\, \mmS_{\alpha}^c)$ for all $S_j\in \mmS_{\alpha,l}$.
\end{proposition}

The results shown above provide a proof of the following lemma.

\begin{lemma}\label{appear}
If a substrate $S_t\in\mmS_{\alpha}^c$ is a  variable in the rational function $r_j^S(\mmS_{\alpha}^c)$ for some $S_j\in \mmS_{\alpha}$, then  there is $S_i\in \mmS_{\alpha}$  and  $S_u\in \mmS_{\alpha}^c$, such that $S_i+S_t$ ultimately reacts to  $S_j+S_u$ via $\mmY_{\alpha}$.
\end{lemma}

After substitution of the value of $S_j$ given in \eqref{Srational} into $Y_k$  \eqref{yelim} we obtain
\begin{equation}\label{ysubst} Y_k= r^Y_k(\mmS_{\alpha}^c) S_1,\end{equation}
 where $r_k^Y$ is either zero or  an S-positive rational function in $\mmS_{\alpha}^c$ with coefficients in $\R(\Con)$. 
If $\widehat{G}_{\mmY_{\alpha}}$ is strongly connected then this function is non-zero.

\medskip
{\bf Conservation laws.} 
The sum of the species concentrations in  $G_{\mmS_{\alpha},\mmY_{\alpha}}$ is conserved. If the total amount  $\oS_1=S_1+\dots+S_{N_{\alpha}}+Y_1+ \dots+ Y_{P_{\alpha}}$  is given, we obtain
$$ \oS_1 = (1+\,r_2^S\!(\mmS_{\alpha}^c)+\dots+r_{N_{\alpha}}^S\!(\mmS_{\alpha}^c) + r_1^Y\!(\mmS_{\alpha}^c)+\dots + r_{P_{\alpha}}^Y\!(\mmS_{\alpha}^c)\,)S_1,$$ 
where the coefficient of $S_1$ is an S-positive element of $\R(\Con\cup\, \mmS_{\alpha}^c)$ and thus,
$$ S_1 = \overline{r}^S_1(\mmS_{\alpha}^c),$$ 
with $\overline{r}^S_1$ an S-positive rational function in $\mmS_{\alpha}^c$ with coefficients in $\R(\Con\cup \{\oS_1\})$.

Further, if $\oS_1>0$ then $S_1\neq 0$ at steady state and $S_1> 0$ for non-negative values of the substrates in  $\mmS_{\alpha}^c$.
This remark and Proposition \ref{snonzero} imply:

\begin{proposition}\label{snonzero2}
A connected component $G_{\mmS_{\alpha,l}}$ of the graph $G_{\mmS_{\alpha}}$  is strongly connected if and only if any steady state solution satisfies $S_j\neq 0$ for all $S_j\in \mmS_{\alpha,l}$, and any total amounts $\oS_l>0$.
\end{proposition}

 By substitution of $S_1$ by $\overline{r}^S_1$, we obtain 
\begin{equation}\label{substcons} Y_k= \overline{r}^Y_k(\mmS_{\alpha}^c):= r^Y_k(\mmS_{\alpha}^c)\overline{r}^S_1(\mmS_{\alpha}^c),\qquad S_j= \overline{r}^S_j(\mmS_{\alpha}^c):= r^S_j(\mmS_{\alpha}^c)\overline{r}^S_1(\mmS_{\alpha}^c)\end{equation}
 with $\overline{r}_k^Y,\overline{r}_j^S$  either zero or S-positive rational functions in $\mmS_{\alpha}^c$ with coefficients in $\R(\Con\cup \{\oS_1\})$.

\begin{proposition}\label{resultS1}
Assume that for each $l=1,\dots,n_{\alpha}$, there exists a spanning tree of $G_{\mmS_{\alpha,l},\mmY_{{\alpha},l}}$  rooted at some species $S_{i_l}$. Then,  equations \eqref{dS2} are satisfied if and only if
$$ S_{j}=r_{j}^S(\mmS_{\alpha}^c) S_{i_l},\qquad S_j\in \mmS_{\alpha,l},$$
where  $r_{j}^S$ is  zero or an S-positive rational function in $\mmS_{\alpha}^c$ with coefficients in $\R(\Con)$. Further,   the conservation law $\oS_{l}=\sum_{S_i\in \mmS_{{\alpha},l}} S_i +\sum_{Y_k\in  \mmY_{{\alpha},l}}Y_k$  is fulfilled if and only if 
\begin{equation}\label{s1elim}
 S_{i_l}= \overline{r}^{S}_{i_l}(\mmS_{\alpha}^c),
 \end{equation}
where $\overline{r}^{S}_{i_l}$ is an S-positive rational function in $\mmS_{\alpha}^c$ with coefficients in $\R(\Con \cup \{\oS_{l}\})$.
\end{proposition}

In Example \eqref{mainex}, the  graph $G_{\mmS_1}$ has two connected components: $S_5$, which does not allow further eliminations, and 
$\xymatrix{
S_1 \ar@<0.3ex>[r]^{\widetilde{b}_{2,1}} & S_2 \ar@<0.3ex>[l]^{\widetilde{b}_{1,2}}}$,
which is strongly connected. Selecting $S_1$ as the non-eliminated species   we obtain
$$S_2= \frac{\widetilde{b}_{2,1}}{\widetilde{b}_{1,2}}S_1 = \frac{ d_{1,2} }{b_{1,4}^2\mu_{2,3}^2S_3}  S_1,\quad Y_1=\frac{ d_{1,2} \mu_{2,3}^1}{b_{1,4}^2\mu_{2,3}^2} S_1 = \frac{ d_{1,2}(b_{1,4}^2+c_{2,1})}{b_{1,4}^2c_{1,2}}S_1, \quad  Y_2=\frac{ d_{1,2} }{b_{1,4}^2} S_1.$$
The total amount equations $\oS_1=S_5+Y_3$ and $\oS_2=S_1+S_2+Y_1+Y_2$ give:
\begin{align*}
\oS_1 & = S_5(1+\mu_{4,5}^3  S_4), &
\oS_2 & =\frac{ d_{1,2} }{b_{1,4}^2}\left ( \frac{1 }{\mu_{2,3}^2S_3}  + \frac{ b_{1,4}^2+c_{2,1}}{c_{1,2}} +1 +\frac{b_{1,4}^2}{d_{1,2}}\right) S_1.  
\end{align*}
Let $\widetilde{r}^S_1(S_3,S_4)=\oS_2 \left ( \frac{1 }{\mu_{2,3}^2S_3}  + \frac{ b_{1,4}^2+c_{2,1}}{c_{1,2}} +1+\frac{b_{1,4}^2}{d_{1,2}} \right)^{-1}$; thus:
\begin{align}
S_1 & = \frac{b_{1,4}^2}{d_{1,2}}\widetilde{r}^{S}_{1}(S_3,S_4), &  S_2&= \frac{\widetilde{r}^S_1(S_3,S_4)}{\mu_{2,3}^2S_3},  & S_5 &= \frac{\oS_1}{1+\mu_{4,5}^3S_4}, \label{mainexY} \\
 Y_1& = \frac{ (b_{1,4}^2+c_{2,1})}{c_{1,2}}\widetilde{r}^S_1(S_3,S_4), &  Y_2&=\widetilde{r}^S_1(S_3,S_4), & Y_3 &= \frac{\mu_{4,5}^3 \oS_1 S_4}{1+\mu_{4,5}^3S_4}. \nonumber
 \end{align}
Thus, all species are given as S-positive rational  functions of $S_3,S_4$ in the coefficient field $\R(\Con\cup \{\oS_1,\oS_2\})$.

\subsection{Steady state equations}
To summarize, at steady state the intermediate complexes $\mmY$ can be expressed as rational functions of the substrates $\mmS$ and therefore eliminated. Further, provided a cut $\mmS_{\alpha}$ exists, the variables $\mmS_{\alpha}$ can be expressed as functions of $\mmS_{\alpha}^c=\mmS\setminus\mmS_{\alpha}$ and therefore also eliminated. For the latter statement, we make use of the conservation laws (with given total amounts) for the species in $\mmS_{\alpha}$ determined by the connected components of $G_{\mmS_{\alpha},\mmY_{\alpha}}$.

Specifically, consider the steady state equations  \eqref{dS2} for  $\mmS_{\alpha}^c$. Substituting the expressions in \eqref{ysubst} and \eqref{Srational} for $\mmY$ and  $\mmS_{\alpha}$ provides  the steady states equations in terms of $\mmS_{\alpha}^c$ and the selected variables $S_{i_l}$ (one for each conencted component of  $G_{\mmS_{\alpha}}$). Using  \eqref{substcons},   the steady states equations are given in terms of $\mmS_{\alpha}^c$ only. Let $\mmS_{\alpha}^c=\{S_{N_{\alpha}+1},\dots,S_N\}$ and  let $\Phi_u(\mmS_{\alpha}^c)=0$ be the equation obtained from $\dot{S}_u=0$ after elimination of $\mmY$ and $\mmS_{\alpha}$ and removal of denominators. The denominators can be chosen to be S-positive and we can multiply the expressions by  the  denominators without changing the \emph{positive} solutions.

Assume that the graph  $G_{\mmS_{\alpha},\mmY_{\alpha}}$ has $n_{\alpha}$ connected components,  and recall that each of them gives rise to only one conservation law (Proposition \ref{maxrank}). Extend the set of conservation laws to a maximal set  of $\dim(\Gamma^{\perp})$ laws.

\begin{theorem}\label{steadystates}
Consider a PTM  system for which there exists a cut $\mmS_{\alpha}$. Further, assume that   each  connected component of $G_{\mmS_{\alpha}}$ admits a rooted spanning tree. If total amounts $\oS_{l}$ are given for the $n_{\alpha}$ connected components of $G_{\mmS_{\alpha},\mmY_{\alpha}}$ and the $\dim(\Gamma^{\perp})-n_{\alpha}$ additional conservation laws, then the non-negative steady states of the system with positive values for all substrates in  $\mmS_{\alpha}^c$ are in one-to-one correspondence with the positive solutions to 
$$\Phi_u(\mmS_{\alpha}^c)=0,\qquad  \oS_{l}=\varphi_l(\mmS_{\alpha}^c)$$ 
for $u=N_{\alpha}+1,\dots,N$ and $l=n_{\alpha}+1,\dots,\dim(\Gamma^{\perp})$.
\end{theorem}
\begin{proof}
We have shown that any non-negative steady state solution with positive values for all substrates in  $\mmS_{\alpha}^c$  must satisfy these equations.  For the reverse, consider a positive solution $s=(s_{N_{\alpha}+1},\dots,s_N)$  to the equations $\Phi_u(\mmS_{\alpha}^c)=0$ and $\oS_l=\varphi_l(\mmS_{\alpha}^c)$. For $i=1,\dots, N_{\alpha}$, define $s_i$ through equation \eqref{s1elim} and $y_k$,  $k=1,\dots,P$, through equation \eqref{yelim}. For positive rate constants and positive total amounts, $s_i,y_k$ are non-negative (because of the S-positivity of the rational functions defining them). By construction these definitions automatically ensure that the conservation laws with total amounts $\oS_{l}$, $l=1,\ldots,N_{\alpha}$,  are satisfied (see Proposition \ref{resultS1}). 

By Proposition \ref{Yelim}, the values $y_1,\dots,y_P$  satisfy \eqref{dY1} for all $k$ and hence the steady state equations of the intermediate complexes are satisfied. By Proposition \ref{resultS1} the values $s_1,\dots,s_{N_{\alpha}}$ satisfy \eqref{dS2}. Since the latter is just \eqref{dS1} after substitution of \eqref{yelim}, we see that \eqref{dS1} holds as well. Since $\Phi_u(\mmS_{\alpha}^c)=0$ is the steady state equation $\dot{S_u}=0$ after substitution of \eqref{yelim} and \eqref{s1elim}, this equation is also satisfied and the same reasoning applies to the equation $\oS_{l}=\varphi_l(\mmS_{\alpha}^c)$, $l> n_{\alpha}$. Thus, $S_i=s_i$ and $Y_k=y_k$ is a solution to the steady state equations and satisfy the conservation laws corresponding to the total amounts $\oS_{l}$. 
\end{proof}

This theorem together with Proposition \ref{yelim}(iii) and Proposition \ref{snonzero2} gives the following corollary.

\begin{corollary}\label{steadystates2} Assume that $\widehat{G}_{\mmY}$ is strongly connected and that for all $S\in \mmS$ there exists a cut $\mmS_{\alpha}$  such that $S\in \mmS_{\alpha}$ and 
 $G_{\mmS_{\alpha}}$  is strongly connected. Then, $S_i=0$ or $Y_k=0$ is not a steady state solution for any $i,k$.  With the notation of Theorem \ref{steadystates}, the non-negative steady states of the system  are in one-to-one correspondence with the non-negative solutions to 
$$\Phi_u(\mmS_{\alpha}^c)=0,\qquad  \oS_{l}=\varphi_l(\mmS_{\alpha}^c)$$ 
for $u=N_{\alpha}+1,\dots,N$ and $l=n_{\alpha}+1,\dots,\dim(\Gamma^{\perp})$.
\end{corollary}

In Example \eqref{mainex}, $\dim(\Gamma^{\perp})-n_{\alpha}=1$ and only one conservation law is  missing, $\oS_3=S_3+S_4+Y_1+Y_2+Y_3$. 
The elimination procedure leads to the steady state equations consisting of $\dot{S}_3=0$ ($\Phi_3$) and $\oS_3$ ($\varphi_3$):
\begin{align*}
0 = \Phi_3(S_3,S_4) & = -b_{1,4}^2\widetilde{r}^S_1(S_3,S_4)  + \frac{b_{3,5}^3\mu_{4,5}^3 \oS_1S_4}{1+\mu_{4,5}^3S_4} \\ 
\oS_3 = \varphi_3(S_3,S_4) & =S_3+S_4+\frac{ b_{1,4}^2+c_{2,1}+c_{1,2}}{c_{1,2}}\,\widetilde{r}^S_1(S_3,S_4)+\frac{\mu_{4,5}^3 \oS_1S_4}{1+\mu_{4,5}^3S_4}. 
\end{align*}
Since the conditions of Corollary \ref{steadystates2} are fulfilled, any non-negative solution of this reduced system provides a non-negative steady state of the PTM system. The steady states of the other species, $S_1,S_2,S_5,Y_1,Y_2,Y_3$, are found  from \eqref{mainexY}. In this specific example, the first equation is easily transformed into a linear equation in $S_3,S_4$, and hence either $S_3$ or $S_4$ can be eliminated as well, providing a polynomial equation in the remaining variables. In this case, S-positivity is not guaranteed.

In the example we intentionally selected $\mmS_{\alpha}$ to have  the highest possible number of elements, since  all these variables are subsequently  eliminated. In Example \eqref{nopart}, the cut $\mmS_{\alpha}=\{E,S_3,S_4\}$ allows us to eliminate three substrates and reduce the steady state equations to a system of three equations in three variables.

In some systems (see e.g.~Section \ref{SigCasc}) there two different cuts $\mmS_{\alpha}$, $\mmS_{\alpha}'$ might exist, such that the union is not a cut, but still all variables in $\mmS_{\alpha}\cup\mmS_{\alpha}'$ can be eliminated. Thus,  more species might  be eliminated if different cuts are considered.

\section{Examples}\label{examples}

\subsection{TG framework}\label{tgframework}
In \cite{TG-rational}, the authors provide a linear elimination procedure for the special case in which the set of substrates is partitioned into two distinct sets. In their context, 
a \emph{PTM system} (here called \emph{TG  system}) consists of three non-empty and disjoint sets of species called enzymes, substrates, and intermediate complexes:
$$\Enz = \{E_{1},\dots,E_{L}\}, \  \Sub = \{S_{1},\dots,S_{N}\}, \  \Int = \{Y_{1},\dots,Y_{P}\},$$
and a set of reactions $\Rct=R_{a} \cup R_{b} \cup R_{c}$ with
\begin{align*}
R_{a} &= \{ E_{i}+S_{j} \xrightarrow{a_{i,j}^k} Y_{k}| (i,j,k)\in I_{a} \} & R_{c} &=  \{Y_{i} \xrightarrow{c_{i,j}} Y_{j} | (i,j)\in I_{c}\} \\
R_{b} &=  \{ Y_{k} \xrightarrow{b_{i,j}^k} E_{i}+S_{j} | (i,j,k)\in I_{b} \} 
\end{align*}
for $I_{a}, I_{b} \subseteq \{1,\dots,L\}\times \{1,\dots,N\}\times \{1,\dots,P\}$ and $I_{c} \subseteq \{1,\dots,P\}^2$, such that
(i) All chemical species are involved in at least one reaction; (ii) For every intermediate complex $Y_{k}$ there is at most one enzyme $E_{\eta(k)}$, such that $(\eta(k),j,k)\in R_{a}\cup R_{b}$ for some $j$;
(iii) If two intermediate complexes  $Y_{k},Y_{v}$ are 1-linked, then $E_{\eta(k)}=E_{\eta(v)}$.
Further, the graph $\widehat{G}_{\mmY}$ and each connected component of the graph $G_{\Sub}$ are required to be strongly connected. In particular, the assumption that $\widehat{G}_{\mmY}$ is strongly connected implies that any $Y_k$ ultimately reacts to $S_i+S_j$ for some $i,j$. This is our Assumption (ii) of a PTM system.

Essentially, they consider post-translational modification systems in which the enzymes are not allowed to be modified. Let  $\mmS=\Sub\cup \Enz$, $\mmS_{\alpha}=\Sub$, and $\mmS_{\alpha}^c=\Enz$. Properties (i)-(iii) imply that $\mmS_{\alpha}$ is a cut. Note that $\mmY_{\alpha}=\mmY_{\alpha}^c=\mmY$. Thus the framework developed here is an extension of the framework developed in \cite{TG-rational}. 

By assumption (iii) the graph $G_{\mmS_{\alpha}^c,\mmY}$ has $L$ connected components that provide $L$ conservation laws for the enzymes: $ \oE_i = E_i + \sum_{k|\eta(k)=i} Y_k  $, for $i=1,\dots,L$. 
With the notation of Lemma \ref{conslaws},   $N_{\alpha}^c=n_{\alpha}^c=L$, $P_{\alpha}^c=0$, so that $N_{\alpha}^c+P_{\alpha}^c-n_{\alpha}^c=0$ and thus a set of independent conservation laws of a TG system can be derived from the non-interacting graphs of $G_{\mmS,\mmY}$. Further,  the form of $R_a$ and $R_b$ ensures that any non-interacting graph contains species either from $\Enz$ or $\Sub$, but not both. Thus, all conservation laws are associated with a connected component either of $G_{\Enz,\mmY}$ or $G_{\Sub,\mmY}$.

It follows that if all intermediate complexes ultimately dissociate into an enzyme and a substrate, and each connected component of $G_{\mmS_{\alpha}}$ admits a rooted spanning tree, then  elimination of the variables in $\mmS_{\alpha}\cup \mmY$ reduces the steady state equations to $L$ equations derived from the total amount of enzymes.

\subsection{Signaling cascades}
\label{SigCasc}
Our setting is well-suited to study elimination of variables in signaling pathways.  Signaling pathways form a special type of PTM systems and an extension of TG systems to include some substrates that also act as enzymes.

\begin{definition}\label{signaling} A \emph{signaling cascade} is a collection of  TG systems  $R^{1},\dots,R^{n}$, with corresponding sets of species  
$$\Enzn^{i} = \{E^{i}_{1},\dots,E^{i}_{L_{i}}\}, \quad \Subn^{i} = \{S^{i}_{1},\dots,S^{i}_{N_{i}}\}, \quad \mmY^{i} = \{Y^{i}_{1},\dots,Y^{i}_{P_{i}}\}$$
and sets of reactions  $\Rct^i=R_{a}^i \cup R_{b}^i \cup R_{c}^i$, for $i=1,\dots,n$, satisfying the following conditions:
\begin{enumerate}[(i)]
\item $(\Enz^{i}\cup \Subn^{i}\cup\, \mmY^i)    \cap (\Enz^{j}\cup \Subn^{j}\cup\, \mmY^j)=\{E^{i+1}_1\} = \{S_{N_i}^i\}$ if $j=i+1$ and it is empty otherwise.
\item For all $i$, each connected component of the graph $G_{\Subn^i}$ admits a spanning  tree rooted at $S_{N_i}^i$. 
\item All intermediate complexes ultimately dissociate into two substrates.
\end{enumerate}

\end{definition}
Condition (i)  implies that a signaling cascade consists of independent TG systems ``joined'' by only one substrate acting as an enzyme in the layer below. This description fits   signaling pathways in which the signal is transmitted downstream. Condition (ii) ensures that the intermediate complexes can be  eliminated. 

Let $N=N_1+\dots+N_n$, $L=L_1+\dots+L_n$ and $ \mmS=\bigcup_i\Enz^{i}\cup \Subn^{i}$.
For each $i$, consider the subset $\Subn^i\subset\, \mmS$.  The associated set of intermediate complexes is  $\mmY_{\Sub^i}=\mmY^i\cup \{Y_k\in \mmY^{i+1}|\, \eta(k)=S_{N_i}^i \}$, and $\Subn^i$ is closed (TG systems do not incorporate reactions $S_u\rightarrow S_j$ among substrates or enzymes). By definition, substrates in $\Subn^i$ do not  interact and thus $\Subn^i$ is a cut.

For simplicity, we assume that the graph $G_{\Subn^i}$ is connected for each $i$.
By Proposition \ref{resultS1}, elimination of the variables in $\Subn^i$ provides 
the steady state relation
$$ S_j^i = r_j^i(\Enzn^i) S_{N_i}^i, \qquad S_j^{i} \in \Subn^i\setminus \{S_{N_i}^i\}.$$ 
By Lemma \ref{appear}, $r_j^i$ depends on the species in $\Enzn^i$ only: if $S_u^i+S_t$ ultimately reacts to $S_j^i+S_r$ for some species $S_u^i$ in $\Subn^i$  and $S_r\in \mmS\setminus \Subn^i$ via $\mmY$ , then  since $S_j^i\neq S_{N_i}^i$,   $S_t=S_r=E_{\eta}^i$ for some $E_{\eta}^i\in \Enz^i$. Further, if $Y_k\in \mmY_{\Sub^i}$, we let $Y_k=r^Y_k(\Enzn^i) S_{N_i}^i$ be the corresponding rational function. 

\medskip
{\bf Conservation laws.} 
Since $G_{\Subn^i}$ is connected and admits a rooted spanning tree, the sum of the species in the graph $G_{\Subn^i,\mmY_{\Subn^i}}$ provides the only conservation law among the species in $\Subn^i\cup \mmY_{\Subn^i}$. Thus,  for each $i$, let a total amount $\oS_i$ be given. We have at steady state
\begin{equation}\label{sigcascsub}
 \oS_i = \sum_{S_j^i \in \Subn^i}  r^i_j(\Enzn^i) S_{N_i}^i  +\sum_{Y_k\in \mmY_{\Sub^i}} r^Y_k(\Enzn^i) S_{N_i}^i. 
 \end{equation}
For $i=n$, $S_{N_n}^n\notin \Enzn^n$, and so $S_{N_n}^n$ is expressed as a rational function in $\Enzn^n$.

Thus, if we let  $\Enzn=\bigcup_{i} \Enzn^i$, we have  that the species in $\mmS\setminus \Enz$ are given as rational functions in $\Enz$   with coefficients  in $\R(\Con)$.  Condition (iii) implies that for  $E\in \Enzn^i\setminus  \{S_{N_i}^i\}$, $\{E\}$ is a cut with associated (connected) graph $G_{E,\mmY_E}$. Thus, if the total amount $\oE$ is provided, the steady states must fulfill the equality
\begin{equation}\label{sigcascenz} 
\oE= E + \sum_{k| E=E_{\eta(k)}} Y_k = E+ \sum_{k| E=E_{\eta(k)}}r^Y_k(\Enzn^i) S_{N_i}^i.
\end{equation}

We conclude that the non-negative steady states of a signaling cascade are solutions to $L$ equations in  $\Enz$ with coefficients in $\R(\Con)$, provided that total amounts for $\Enz$ are given; that is, $\oS_1,\dots,\oS_{n-1}$ for the enzymes $S_{N_i}^i$, \eqref{sigcascsub} and $\oE^i_{\eta}$ for  $E^i_{\eta}\in\Enzn\setminus  \{S_{N_1}^1,\dots,S_{N_{n-1}}^{n-1}\}$, \eqref{sigcascenz}.

Note that the number of conservation laws obtained in this way is $m=\sum_i L_i+1$ (remember $\oS_n$). Let $\epsilon=1$ if $n$ is even and $0$ otherwise, and  let $\epsilon^c=1-\epsilon$.   The cuts provide all conservation laws: The graph associated to the cut $$\mmS_{\alpha} = \bigcup_{i \textrm{ even}} \Subn^{i} \cup \bigcup_{i \textrm{ odd}} \Enzn^i$$ has $n_{\alpha}=\epsilon+\sum_{i\textrm{ odd}} L_i$ connected components and thus,  $n_{\alpha}^c=\epsilon^c+\sum_{i\textrm{ even}} L_i$. We have $N_{\alpha}=\sum_{i\textrm{ odd}} L_i+\sum_{i\textrm{ even}} (N_i-1) + \epsilon$, and
$N_{\alpha}^c=\sum_{i\textrm{ even}} L_i+\sum_{i\textrm{ odd}} (N_i-1) + \epsilon^c$.  
Further, $\mmY_{\alpha}=\mmY$, so that $P_{\alpha}^c=0$.  

Let $\gamma=N_{\alpha}^c-n_{\alpha}^c=\sum_{i\textrm{ odd}} (N_i-1)$. By Lemma \ref{conslaws},
if there are $\gamma$ independent terms in $\mmS_{\alpha}^c\cap \Gamma$, then all conservation laws come from non-interacting graphs. By hypothesis, for $i$ even, the graph $G_{\Subn^i}$ has a spanning tree rooted at some node $S_j$. This means that for every $S_u\neq S_j$ in $\Subn^i$, there exists a directed path 
$\xymatrix@C=10pt{S_u \ar@{->}[r] & S_{k_1} \ar@{->}[r] & \dots \ar@{->}[r] & S_{k_r} \ar@{->}[r] & S_j }$.
By the conditions of a TG system and Lemma \ref{block}, an edge $S_{k_v}\rightarrow S_{k_s}$ implies that there is either a reaction $S_{k_v}\rightarrow S_{k_s}$, or $E+S_{k_v}$ ultimately reacts to $E+S_{k_s}$ via $\mmY$. In either case, we see that $S_u-S_j\in \mmS_{\alpha}^c\cap \Gamma$ for all  $S_u\neq S_j$ in $\Subn^i$, implying that there are indeed  $\gamma$ independent vectors in $\mmS_{\alpha}^c\cap \Gamma$.

\subsection{Biological examples}

\hspace{.1cm}
\medskip

{\bf MAPK signaling cascade.}
We consider the first two layers of the MAPK cascade: a two-layer cascade with one-site modification in the first layer and two-site modifications in the second layer. In the latter, dephosphorylation is considered sequential but this is not the case for phosphorylation  \cite{Markevich-mapk}.

 The reactions of the system in the first layer are
 $$ \xymatrix@C=15pt{
E+S_{0}^1 \ar@<0.5ex>[r] & Y_{1}^1 \ar[r] \ar@<0.5ex>[l] & E+ S_{1}^1} \qquad \xymatrix@C=15pt{F_{1}+S_{1}^1 \ar@<0.5ex>[r] & Y_{2}^1 \ar[r]  \ar@<0.5ex>[l] & F_{1}+ S_{0}^1}$$ 
accounting for phosphorylation and dephosphorylation, respectively, via a Michaelis-Menten mechanism. 
In the second layer we have the phosphorylation reactions 
 
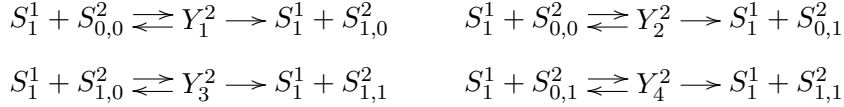
\begin{align*}
 \xymatrix@C=15pt{
S_{1}^1+S_{0,0}^2 \ar@<0.5ex>[r] & Y_{1}^2 \ar[r]  \ar@<0.5ex>[l] & S_{1}^1+ S_{1,0}^2 }\qquad \xymatrix@C=15pt{
S_{1}^1+S_{0,0}^2 \ar@<0.5ex>[r] & Y_{2}^2 \ar[r]  \ar@<0.5ex>[l] & S_{1}^1+ S_{0,1}^2 } \\
\xymatrix@C=15pt{
S_{1}^1+S_{1,0}^2 \ar@<0.5ex>[r] & Y_{3}^2 \ar[r]  \ar@<0.5ex>[l] & S_{1}^1+ S_{1,1}^2 }\qquad \xymatrix@C=15pt{
S_{1}^1+S_{0,1}^2 \ar@<0.5ex>[r] & Y_{4}^2 \ar[r]  \ar@<0.5ex>[l] & S_{1}^1+ S_{1,1}^2 } 
\end{align*}
Dephosphorylation proceeds sequentially in the following way:
$$\xymatrix@C=15pt{
F_2 +S_{1,1}^2 \ar@<0.5ex>[r] & Y_{5}^2 \ar[r]  \ar@<0.5ex>[l] & F_2+ S_{1,0}^2 }\qquad \xymatrix@C=15pt{
F_2+S_{1,0}^2 \ar@<0.5ex>[r] & Y_{6}^2 \ar[r]  \ar@<0.5ex>[l] & F_2+ S_{0,0}^2 } $$ 

The sets of enzymes are $\Enzn^1=\{E,F_1\}$, $\Enzn^2=\{S_1^1,F_2\}$. The sets of substrates are $\Subn^1=\{S_0^1,S_1^1\}$, $\Subn^2=\{S_{0,0}^2,S_{1,0}^2,S_{0,1}^2,S_{1,1}^2\}$. The sets of intermediate complexes are $\Intn^1=\{Y_1^1,Y_2^1\}$, $\Intn^2=\{Y_1^2,Y_2^2,Y_3^2,Y_4^2,Y_5^2,Y_6^2\}$.
We have $\Enzn^2\cap \Subn^1=\{S_1^1\}$, so that the modified substrate in the first layer is a kinase of the next layer. The superindex denotes the layer, while the subindex  denotes  phosphorylation state (the presence of the phosphate group is represented by $1$).

The components of the graph $\widehat{G}_{\mmY}$ are each of the intermediate complexes and are thus strongly connected. The graphs $\widehat{G}_{\Subn^1}$ and $\widehat{G}_{\Subn^2}$ are 
$$
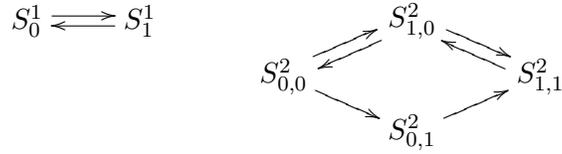
\xymatrix{S_0^1 \ar@<0.4ex>[r] & S_1^1  \ar@<0.4ex>[l]}\qquad\quad \xymatrix@R=1pt{   & S_{1,0}^2 \ar@<0.5ex>[dr]  \ar@<0.5ex>[dl] \\  S_{0,0}^2 \ar[dr] \ar@<0.5ex>[ur] & & S_{1,1}^2  \ar@<0.5ex>[ul] \\  & S_{0,1}^2  \ar[ur]}$$
which  are also strongly connected. 
The conservations laws (all derived from non-interacting graphs) are 
\begin{align*}
\oE &= E+Y_1^1 & \oS_1 &= S_0^1+S_1^1 + Y_1^1 + Y_2^1 + Y_1^2+Y_2^2+Y_3^2+Y_4^2 \\
\oF_1 &= F_1+Y_2^1  & \oS_2 &= S_{0,0}^2 + S_{1,0}^2 + S_{0,1}^2 + S_{1,1}^2 +  Y_1^2+Y_2^2+Y_3^2+Y_4^2 + Y_5^2+Y_5^2\\
\oF_2 &= F_2 + Y_5^2 + Y_6^2
\end{align*}
Therefore, if total amounts are provided, then the steady states of the two-layer cascade are found as solutions to a system of four polynomial equations in four variables, namely $E,F_1,F_2,S_1^1$.

\medskip
{\bf Receptor protein-tyrosine kinase.}
Receptor protein-tyrosine kinases (RPTK) 
are cell surface receptors  linked to enzymes that phosphorylate their substrate proteins in tyrosine residues. The common mechanism for their activation is autophosphorylation following ligand-induced dimerization \cite[$\S$15]{cell}. The phosphorylated receptor serves as binding site to downstream signaling molecules, such as SH2 domain containing proteins. Further, the receptor can be dephosphorylated by several protein tyrosine phosphatases (PTP) \cite{ostman}. 

A simple model describing the phosphorylation state of an RPTK  is:
$$ 
\xymatrix@C=15pt{  2R_0 \ar@<0.5ex>[r] & Y_1 \ar[r] \ar@<0.5ex>[l] & 2R_1 & S + R_1 \ar@<0.5ex>[r] & Y_2 \ar@<0.5ex>[l]  &   F + R_1 \ar@<0.5ex>[r] & Y_3 \ar[r] \ar@<0.5ex>[l] &  F+R_0   }$$
where $R_0,R_1$ stands for the unphosphorylated and phosphorylated RPTK respectively, $S$ is a protein binding $R_1$, and $F$ is a PTP. 

We have  $\mmS=\{R_0,R_1,S,F\}$ and $\mmY=\{Y_1,Y_2,Y_3\}$.  Note that $\mmS_{\self}=\{R_0,R_1\}$ are the self-interacting substrates and thus cannot  be part of a cut.
First of all, the intermediate complexes $Y_k$ can be eliminated in terms of $\mmS$. 
The graph $G_{\mmS,\mmY}$ is
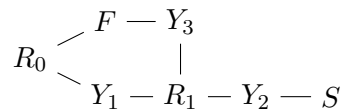
\begin{center}
\begin{tikzpicture}
\node (R0) at (1,0) {$R_0$};
\node (F) at (2,0.5) {$F$};
\node (Y1) at (2,-0.5) {$Y_1$};
\node (Y3)at (3,0.5) {$Y_3$};
\node (R1)at (3,-0.5) {$R_1$};
\node (Y2)at (4,-0.5) {$Y_2$};
\node (S)at (5,-0.5) {$S$};

\path[draw] (R0) to (Y1);
\path[draw] (R0) to (F);
\path[draw] (R1) to (Y1);
\path[draw] (Y3) to (F);
\path[draw] (S) to (Y2);
\path[draw] (R1) to (Y2);
\path[draw] (R1) to (Y3);

\end{tikzpicture}
\end{center}
The non-interacting graphs provide two conservation laws: $\oF = F + Y_3$, and $\oS = S+Y_2$, associated to the cut  $\mmS_{\alpha}=\{F,S\}$. Thus, the substrates $F,S$ can be eliminated, in fact from the conservation laws.
We conclude that at steady state all species are described   as rational functions of $R_0,R_1$ and the non-negative steady states are in one-to-one correspondence with the non-negative solutions to the equations corresponding to $\dot{R}_0$ and the remaining conservation law $\overline{R}=R_0+R_1+2Y_1+Y_2+Y_3$.

\section*{Acknowledgments}
EF is supported by a postdoctoral grant from the ``Ministerio de Educaci\'on'' of Spain and the project  MTM2009-14163-C02-01 from the ``Ministerio de Ciencia e Innovaci\'on''.  CW is supported by the Lundbeck Foundation, Denmark and the Leverhulme Trust, UK. 
 

\end{document}